 \documentclass[journal]{IEEEtran}

\usepackage{circledsteps}
\usepackage{hyperref}
\hypersetup{
 colorlinks,
 citecolor=gray,
 linkcolor=purple,
}

\usepackage{url}
\usepackage{mathtools}
\usepackage{amsmath}
\usepackage{amsthm}
\usepackage{xspace}
\usepackage{hhline}
\usepackage{balance}
\usepackage{algorithm}
\usepackage{algpseudocode}
\usepackage{comment}
\usepackage{thm-restate}
\usepackage{dsfont}
\usepackage{float}

\newcommand{\msn}{\textsc{comp-sys}\xspace}
\newcommand{\csn}{\textsc{BvN-sys}\xspace}
\newcommand{\rsn}{\textsc{rr-sys}\xspace}

\DeclareMathOperator*{\argmin}{arg\,min}

\newcommand{\salg}{\ifmmode \mathcal{A}_{rr}\xspace \else $\mathcal{A}_{rr}\xspace$ \fi}
\newcommand{\stalg}{\ifmmode \mathcal{A}_{\mathrm{top}} \else $\mathcal{A}_\mathrm{top}\xspace$ \fi}
\newcommand{\spalg}{\ifmmode \mathcal{A}_{\mathrm{pkg}} \else $\mathcal{A}_\mathrm{pkg}\xspace$ \fi}

 \renewcommand{\skew}{ \phi\xspace}

\newcommand{\para}[1]{\vspace{2pt} \noindent \textbf{#1:}\xspace}
\newcommand{\rrec}{R_{r}}
\newcommand{\crec}{R_{b}}

\newcommand{\dct}{\ifmmode \mathrm{\Delta} \else \Delta \fi} 
\newcommand{\dctrot}{\Delta}
\newcommand{\dctda}{\Delta}
\newcommand{\dctmix}{\Delta}
\newcommand{\dctSchedule}{\ifmmode \mathrm{Time} \else \textrm{Time}\xspace \fi }
\newcommand{\sys}{\textsc{sys}}

\newcommand{\rr}{\textsc{rr}}
\newcommand{\direct}{\textsc{Direct}}
\newcommand{\rrbvn}{\textsc{MulP}\xspace}
\newcommand{\algOnePerm}{\textsc{Perm}\xspace} 
\newcommand{\upper}{\textsc{Upper}\xspace}
\newcommand{\upperU}{\ifmmode \textsc{Upper}^{+} \else \textsc{Upper}$^+$\xspace \fi }
\newcommand{\lowerRsn}{\textsc{Opt}}

\newcommand{\algPivot}{\textsc{Pivot}\xspace} 
\newcommand{\algPivotU}{\textsc{Pivot}$^+$\xspace} 
 
\newcommand{\cycleSched}{\mathcal{C}} 
\newcommand{\cycleSchedElem}{C} 
\newcommand{\configElem}{C} 
\newcommand{\bvn}{\textsc{BvN}\xspace}
\newcommand{\eff}{\eta} 

\newcommand{\stalgmsn}{{ \textsc{Pivot}^{}_T} }
\newcommand{\spalgmsn}{{ \textsc{Pivot}^{}_P}}
\newcommand{\stalgmsnU}{{ \textsc{Pivot}^+_T}}
\newcommand{\spalgmsnU}{{ \textsc{Pivot}^+_P}  }

 \newcommand{\rotornet}{\emph{RotorNet}\xspace} 
\newcommand{\ps}{\ifmmode \mathcal{S}_{\mathrm{pkg}} \else $\mathcal{S}_{\mathrm{pkg}}\xspace$ \fi}
\newcommand{\ts}{\ifmmode \mathcal{S}_{\mathrm{top}}\else $\mathcal{S}_{\mathrm{top}}\xspace$ \fi}
\newcommand{\totalDemandMat}{\hat{M}}

\newcommand{\ditPS}{T_{\mathrm{pkg}}}

\newcommand{\mdl}{m^{(dl)}}
\newcommand{\Mdl}{M^{(dl)}}
\newcommand{\moh}{m^{(1h)}}
\newcommand{\Moh}{M^{(1h)}}
\newcommand{\mth}{m^{(2h)}}
\newcommand{\Mth}{M^{(2h)}}
 \newcommand{\mdlSum}{\hat{m}^{(dl)}}
\newcommand{\mohSum}{\hat{m}^{(1h)}}

\newcommand{\mthSum}{\hat{m}^{(2h)}}

\newcommand{\inActCell}{\omega^\emptyset} 
\newcommand{\inActCellMax}{\omega } 
\newcommand{\largeLoad}{c_l}
\newcommand{\largeRatio}{t_l}
\newcommand{\numberFlows}{n_f}
 \newcommand{\numberFlowsSmall}{n_s}
 \newcommand{\numberFlowsLarge}{n_l}
\DeclarePairedDelimiter{\ceil}{\lceil}{\rceil}

\def\sentry{\tau}
\def\TM{\mathrm{TM}}

\newcommand{\card}[1]{\lvert #1 \rvert}
\newcommand{\Wcard}[1]{\lvert\lvert #1 \rvert\rvert_{\scriptscriptstyle \Sigma}}
\newcommand{\Forbcard}[1]{\lvert\lvert #1 \rvert\rvert_{\scriptscriptstyle  F}}

\newtheorem{observation}{Observation}
\newtheorem{claim}{Claim}
\newtheorem{definition}{Definition}

\newtheorem{theorem}{Theorem} 
\newtheorem{corollary}[theorem]{Corollary}
\newtheorem{lemma}[theorem]{Lemma}

\newcommand{\khen}[1]{\textcolor{blue}{CG: #1}}

\newcommand{\lp}[1]{#1}

\makeatletter

\patchcmd{\maketitle}
	{\@maketitle}
	{\@maketitle\vspace{-0em}}
	{}
	{}
\makeatother

\usepackage{subfiles} 

\begin{document}

\title{Integrated  Topology and Traffic Engineering for Reconfigurable Datacenter Networks}

\author{
    \IEEEauthorblockN{Chen Griner and Chen Avin}\\
    \IEEEauthorblockA{{School of Electrical and Computer Engineering, Ben-Gurion University of the Negev
    \\ griner@post.bgu.ac.il } \IEEEauthorblockA{avin@bgu.ac.il}
}}

 \maketitle
\begin{abstract}
The state-of-the-art topologies of datacenter networks are fixed, based on electrical switching technology, and by now, we understand their throughput and cost well. 
For the past years, researchers have been developing novel optical switching technologies that enable the emergence of reconfigurable datacenter networks (RDCNs) that support dynamic psychical topologies.
The art of network design of dynamic topologies, i.e., 'Topology Engineering,' is still in its infancy. 
Different designs offer distinct advantages, such as faster switch reconfiguration times or demand-aware topologies, and to date, it is yet unclear what design maximizes the throughput.

This paper aims to improve our analytical understanding and formally studies the throughput of reconfigurable networks by presenting a general and unifying model for dynamic networks and their topology and traffic engineering.
We use our model to study demand-oblivious and demand-aware systems and prove new upper bounds for the throughput of a system as a function of its topology and traffic schedules.

Next, we offer a novel system design that combines both demand-oblivious and demand-aware schedules, and we prove its throughput supremacy under a large family of demand matrices.
We evaluate our design numerically for sparse and dense traffic and show that our approach can outperform other designs by up to $25\%$ using common network parameters.
\end{abstract}

\section{Introduction}

Datacenter networks (DCNs) are 
an essential infrastructure for modern internet applications, which are at the base of our digital lives.  Large-scale web services, cloud computing, and the recent proliferation of large language models (LLMs) are just a few of the typical uses of datacenters.
The use of different online platforms and services is expected to increase even further in the coming years, resulting in an explosive increased demand for datacenters \cite{dataCenterGrowth2023URL,Open_Compute_Project2022}. 
DCNs need to improve their performance by adopting new ideas and technologies. 
Traditionally, DCNs have been connected through electrical packet switches in a wired static topology, e.g., fat-trees, clos networks \cite{closAl2008scalable,liu2013f10,jupiterSingh2015}, and expanders \cite{xpanderKassing2017beyond}. In these DCNs, since the topology is static, it needs to function adequately under any traffic pattern, and, therefore, the topology is optimized for a worst-case traffic pattern and is oblivious to the actual traffic. 

However, electrical packet-switched networks are just one technology. Recent years have given rise to optical circuit switches (OCSes), a newer technology with new possibilities \cite{wang2009your,poutievski2022jupiter,hall2021survey}.  
OCS and other optical technologies, such as those based on wavelength gratings as in Sirius \cite{ballani2020sirius}, enable dynamic reconfigurable network topologies via optical circuit switching.
Using these optical switches, the topology can change dynamically in relatively short periods and adapt to different traffic patterns, thus facilitating the emergence of reconfigurable datacenter networks (RDCNs)\cite{wang2009your,avin2019toward}. 
RDCNs require what was recently termed \emph{``Topology engineering''} \cite{poutievski2022jupiter}: a new dimension of network design that controls the dynamic topology of the network.
Essentially, reconfigurable topologies allow flows to route via shorter paths, increasing throughput by saving network capacity at the cost of a temporal penalty. In other words, these systems can save \emph{bandwidth tax} resulting from multi-hop routing at the cost of a \emph{latency tax} resulting from topology reconfiguration time \cite{mellette2017rotornet,addanki2023mars,griner2021cerberus}. 
\begin{figure}[t]
  \begin{centering}
  \begin{tabular}{c}
 \includegraphics[width=0.80\columnwidth]{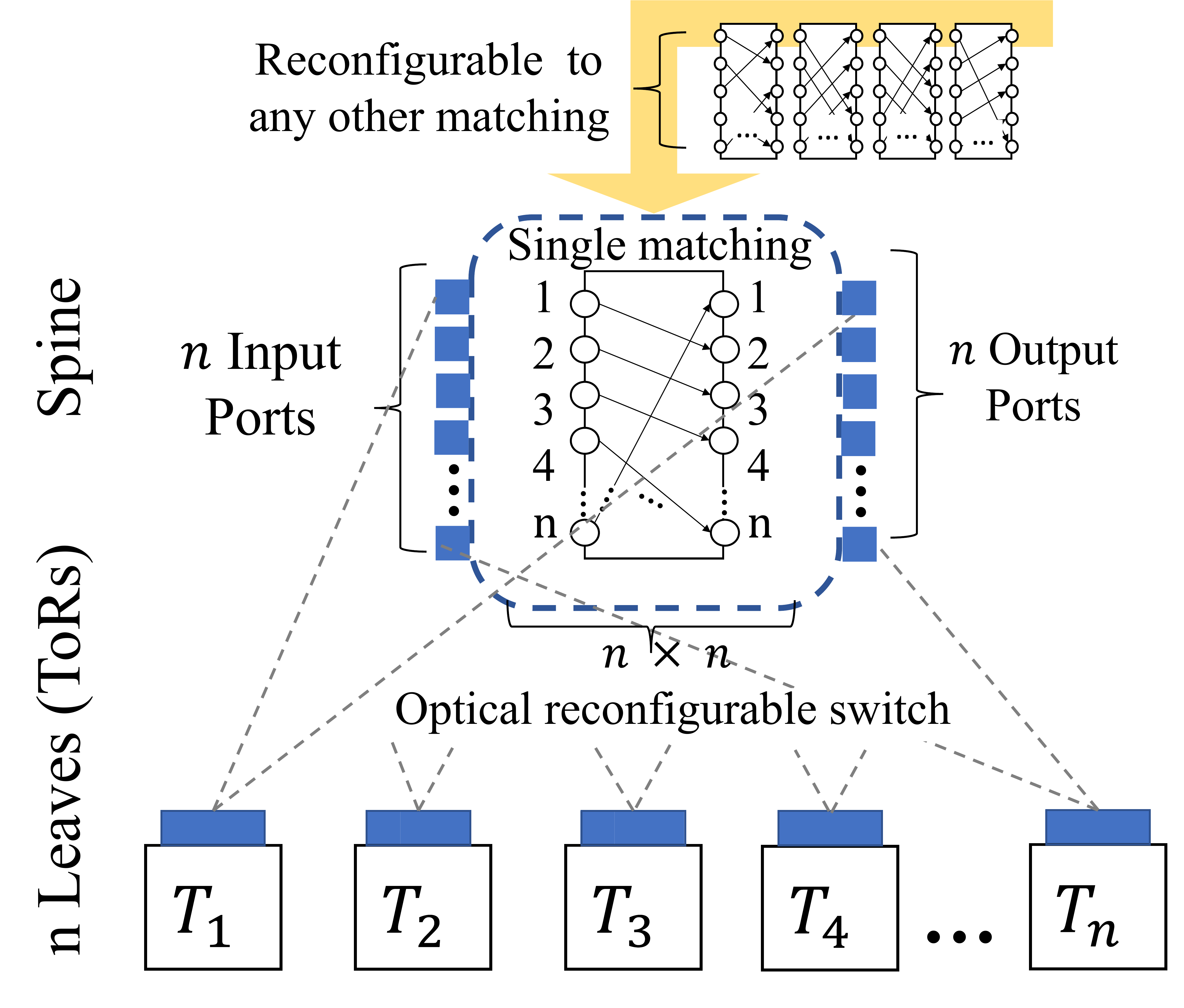} 
  \end{tabular}
   \caption[A schematic view of our simplified network model]{
  A schematic view of our simplified network model. 
  A two-layer leaf-spine architecture. The spine layer consists of a single switch that is able to adapt and change its matching.  The dashed lines represent unidirectional links.
   }
    \label{fig:OrthNetModelIntro}
  \end{centering}
 
\end{figure}

In this paper, we rigorously study networked systems with dynamic topologies, exploring the tradeoffs between different tax types, with the goal of increasing the throughput for these networks.
To do so, we use the most basic and simple model of a reconfigurable datacenter network 
which is illustrated schematically in \autoref{fig:OrthNetModelIntro}.
Similar to previous works \cite{mellette2020expandingOpera,mellette2017rotornet,griner2021cerberus,ghobadi2016projector}, we assume a network based on a two-layer, leaf-spine topology. The leaf layer has $n$  nodes. Each node represents a top-of-rack (ToR) switch. In turn, the nodes are connected using a \emph{single}  
optical spine switch with $n$ input and $n$ output ports of equal capacity. At any given time, this switch can hold any directed \emph{matching} between its input and output ports, creating a directed topology (i.e., graph) between the set of ToRs, $T_1,...,T_n$ (the nodes).
Since at any given time the switch implements a single matching, any node could be connected to any other node, but only to one.
We assume that only a single hop routing is allowed on each matching. However, in our model, data may be sent to a node in one matching and then forwarded to its destination in a different matching. This enables the topology to simulate networks of arbitrary degree over time. 
A similar abstraction to a single optical switch was also made in other works like Helios \cite{farrington2010helios}, Sirius \cite{ballani2020sirius}, and others \cite{wang2009your}. We further discuss our model in \autoref{subsec:DemandMatrixNetworkModel}
The reconfiguration process of a switch, from one matching to another, has a temporal penalty, termed the \emph{reconfiguration time}. Usually, depending on the system type, the reconfiguration time is given as a system parameter.  The sequence of matching that a system implements is denoted as \emph{topology schedule}, and it explicitly defines the dynamic topology of the system. 
The goal of our work is to maximize the throughput of a system; this can be done by the topology schedule, which determines the dynamic topology and its properties, like the path length and the number of reconfigurations.
In our model, multi-hop routing is the main contributor to bandwidth tax, and reconfiguration time serves as the main contributor to latency tax. The networked system's throughput can be increased by reducing the "payment" for both or either tax types. 
The main challenges are to find the best tradeoff between these taxes, and understand what is the lower bound of the achievable throughput for a given system.

There are two basic approaches to topology scheduling (and, more generally,  topology engineering): 
\emph{demand-aware} and \emph{demand-oblivious} topology designs ~\cite{avin2019toward}. 
The two design approaches use different strategies to achieve lower ``taxation''.  
In a nutshell, demand-aware topologies or what may also be called \emph{adaptive topologies} are dynamic networks that adjust the topology to the actual demand, optimizing the topology to particular traffic patterns by rewiring connections.  
On the contrary, demand-oblivious topologies~\cite{hall2021survey} 
present the same topology to any traffic patterns, oblivious to the actual demand.  Note, however, that demand-oblivious systems can be dynamic (e.g., \rotornet \cite{mellette2017rotornet}) or static (e.g., Clos \cite{closAl2008scalable}).  
In this paper, we study both demand-oblivious and demand-aware dynamic topologies.

Typical examples of demand-aware topologies are systems like Eclipse \cite{bojja2016costlyEclipse}, projecTor \cite{ghobadi2016projector}, Morida  \cite{moridafarrington2013multiport}, and others \cite{solsticeliu2015scheduling,hamedazimi2014firefly,zerwas2023duo}. In their approach to topology engineering, the topology can form an arbitrary schedule of matchings where all or most traffic flows are served in a direct connection, i.e., single hop, thus paying no bandwidth tax. 
Some of the best-known approaches for demand-aware topology scheduling (e.g., Eclipse \cite{bojja2016costlyEclipse}, and Solstice  \cite{solsticeliu2015scheduling} are based on Birkhoff von Neumann (BvN) decomposition \cite{frank1956algorithm,birkhoff1946tres} of the demand matrix.
The BvN decomposition is a well-known and widely studied method used to decompose a double stochastic matrix into a set of matchings or permutation matrices, see preliminaries in \autoref{sec:Preliminaries}). 
We denote systems from this family as Birkhoff von Neumann systems or \csn.

An OCS switch can implement the schedule generated using BvN by reconfiguring itself to each permutation matrix in the decomposition in some order. Each reconfiguration incurs a delay in the operation of the OCS during which it cannot transmit data, i.e., the reconfiguration time. 
Bounds on the reconfiguration time are varied.
Most \csn type systems need to reconfigure physical connections (e.g., move mirrors), which can result in a delay on the scale of \emph{micro-seconds} ~\cite{memesF20011296}. Furthermore, they need to calculate a novel schedule, which may be computationally expensive. Indeed, improving the run time of the BvN decomposition is the subject of some recent works \cite{bojja2016costlyEclipse,livshits2018lumos,valls2021birkhoff}.

In contrast to the demand-aware topologies we just described, demand-oblivious system designs can achieve much faster reconfiguration times (e.g., Sirius \cite{ballani2020sirius}, \rotornet \cite{mellette2017rotornet}, and Opera \cite{mellette2020expandingOpera}). 
These networks achieve \emph{nano-second} scale reconfiguration times and pay very little in the form of latency tax \cite{addanki2023mars,griner2021cerberus}.
The topology schedule approach in these designs is to use a set of \emph{predefined matchings} and quickly 'rotate' between them in a \emph{round-robin} and periodic manner, which does not change.
This leads to faster reconfiguration since there is no need to calculate a new topology schedule for each demand pattern.
In these systems, the union of the schedule's matchings results in a complete graph\footnote{Some recent work proposed round-robin systems that emulate other regular graphs like expanders, in the current work, we assume they emulate a complete graph.}. 
All nodes will be connected directly once in each round, or \emph{cycle}. 
Due to the round-robin nature of the operation of systems such as \rotornet and Opera, in this paper, we will denote systems from this family as round-robin systems or, \rsn.
Since reconfiguration time is lower in round-robin systems and the topology configuration is demand-oblivious, one avenue left to improve the throughput of such systems is via \emph{traffic engineering}. 
In particular, demand-oblivious topology schedules typically use two-hop paths for load balancing trading bandwidth tax to help improve the system's throughput. For example, Valiant load balancing \cite{valiant1981universal} as used in Sirius \cite{ballani2020sirius}, or RotorLB, used in \rotornet  \cite{mellette2017rotornet}. 
To implement traffic engineering over a topology schedule, we define a \emph{traffic schedule}, which determines which traffic (e.g., packets) will be sent over which link and at which time. 
We further note that a traffic schedule could be either oblivious (as oblivious routing \cite{azar2003optimalObliviousRouting}) or adaptive (as adaptive routing \cite{lang2001analysisobliviousAdaptiveRouting}). For example, Sirius \cite{ballani2020sirius} uses oblivious routing over oblivious topology while \rotornet  \cite{mellette2017rotornet} and Opera \cite{mellette2020expandingOpera} uses adaptive routing over oblivious topology.

In this paper, we show that to improve the throughput, we should consider \emph{composite} systems (i.e., hybrid topology systems), that is, topology schedules that can combine \emph{both} both round-robin (demand-oblivious) and BvN (demand-aware) designs into a single schedule. We denote such a system as \msn. Recently,  Cerberus~\cite{griner2021cerberus} showed that a hybrid approach could help to improve the throughput in some important cases. In the current work, we provide a more accurate, in-depth analysis of the bounds on the performance of hybrid systems in terms of demand completion time.
We note that since a composite system uses both round-robin and BvN components, we require both a traffic schedule and a topology schedule to define it fully. We further elaborate on this point in the model section of the paper. 

Formally, we study the demand completion time (DCT) of different systems (i.e., \rsn, \csn, \msn), where the DCT of a system is defined as the completion time of the \emph{wors-case} demand matrix for that system. The DCT of a system guarantees an upper on the completion time for \emph{any} demand matrix. Since the DCT is inversely proportional to the throughput \cite{jyothi2016measuring,griner2021cerberus,addanki2023mars}, an upper bound on the DCT of a system provides a lower bound on the system's throughput for \emph{any} demand matrix. In particular, a system with a better DCT than another system, will also have better throughput.
Furthermore, we conjecture that composite systems have higher throughput than \rsn and \csn type systems. 
From a wider perspective, we would like to understand better the trade-offs between different systems and the benefits of composite systems. To this end, our main contributions are the following\\
\para{Problem Formalism}
We formally define a system DCT (\autoref{def:sysDCT}) as a function of its topology and traffic scheduling.  
The system DCT provides an upper bound on the completion time of the system for any demand matrix. 
We claim that the system DCT is an important metric of network performance. Indeed, we argue that a system with a lower system DCT is superior to other systems. \\
\para{Analysis of \rsn and \csn DCTs}
We perform an in-depth analysis of the system DCT of \rsn and \csn.  In particular, we provide the system DCT for \csn in \autoref{obs:main:DAsysDCTBound}, and for \rsn in \autoref{thm:main:RRsysDCT} for general double stochastic matrices. To do so, we expand previous results and offer new insights and novel bounds on the performance of round-robin systems. \\
\para{Novel Algorithms and Analysis of \msn DCT}
We propose a novel algorithm for \msn on an input of general double stochastic matrices and give detailed results for the system-DCT of a composite system on a specific subset of double stochastic matrices in \autoref{thm:mainSystemDCTreslut}.  Our result shows that, under mild assumptions, the system DCT of \msn is better (lower) than either \rsn or \csn. \\
We note, however, that the general case of double stochastic matrices remains an open question.\\
\para{Empirical evaluation}
We empirically evaluate our algorithm's performance on known traffic models using a set of plausible parameters. We show that the \msn system is superior, by more than $20\%$, to any of the other systems, \rsn and \csn, which represent state-of-the-art systems.

The rest of this paper is organized as follows: We begin by explaining our model and problem definition in \autoref{sec:modelAndProbDefinition}, Next, we derive DCT bounds for two existing and known systems in \autoref{sec:NetwrokDesignsCSN} and than in \autoref{sec:NetwrokDesignsRSN}. We present the composite system in \autoref{sec:msn}. We analyze a case study for the family of matrices $M(v,u)$ in \autoref{sec:thecasestudy}. We present an empirical evaluation of our systems in \autoref{sec:Empirical} and the related work in \autoref{sec:relwork}. We conclude with a discussion in Section \ref{sec:discution}.
Due to space constraints, some of the technical results are moved to the Appendix.


\section{Demand Completion Times (DCT): model and problem definition }\label{sec:modelAndProbDefinition}

In this section, we discuss and define the problem of \emph{demand completion time} (DCT) minimization in 
reconfigurable dynamic networks.
We start by defining our problem's input, a demand matrix $M$, representing our network traffic. In turn, we define two types of schedulers: a \emph{topology scheduler} and a \emph{traffic scheduler} (where we consider the special case that allows at most two hop paths).
Both schedulers, together, can formally describe the forwarding of traffic in a dynamic network and the length of the topology scheduler (in seconds), which in turn defines the matrix DCT. Finally, In \autoref{subsec:problemdefinition}, we define the DCT minimization problem as finding the best schedulers for a given system, where each such scheduler's DCT is defined by its worst-case demand matrix.
 
\subsection{Preliminaries  }\label{sec:Preliminaries}
We first recall some relevant, important definitions. 

\para{Doubly stochastic matrix}\label{def:Doublystochasticmatrix}
A square matrix $M$ of non-negative numbers is a \emph{doubly stochastic matrix} (aka bistochastic) if the sum of the cells $M[i,j]$ in each row and column is equal to one, i.e., $\sum_i M[i,j]=\sum_j M[i,j]=1$.

\para{Saturated matrix and scaled matrix}
A $\lambda$ \emph{scaled} doubly stochastic matrix is a matrix where the sum of each row is $\lambda$.
In network traffic, a demand matrix $M$ is called \emph{saturated} if it is $\lambda$ scaled and $\lambda=r$, where $r$ is the capacity of all nodes, i.e., the sum of each row and column, will be $r$.

 \para{Permutation matrix}
 A \emph{permutation matrix} $P$ is a   double stochastic matrix where each row and column has exactly one cell with a value of $1$, and zero otherwise.
 We denote the set of all derangement permutation matrices (without fix points) as $P_\pi$.  
  
\para{The weight of a matrix and a matrix set}
For a demand matrix $M$ we denote its \emph{weight}, $\Wcard{M}$, as the sum of all its elements, 
 $\Wcard{M}=\sum_{ij} M[i,j]$. 
For a set $\overline{M} = \{m_1, \dots, m_v \}$ of such matrices 
we define $\Wcard{\overline{M}}=\sum_{i=1}^v  \Wcard{m_i}$. 

\para{Birkhoff-von Neumann decomposition}
The Birkhoff-von Neumann  theorem \cite{birkhoff1946tres}  (or $\bvn$), states that given a doubly stochastic matrix $M$, there exists a set of $v \le n^2$  positive coefficients $\beta_i>0$,  with the property $\sum_{i=1}^v \beta_i=1$, and a set of $v$ permutation matrices $P_i$, such that their sum equals the original matrix $M$, that is,
   $\sum_{i=1}^v  \beta_i\cdot P_i= M$.

\para{The  M(v) and the M(v,u) Matrices}\label{subsec:MvuMatrixDef}
Along the paper, we use the following matrices.
Let $M(v)$ be defined as a double stochastic matrix where each row and column has $v$ non-zero, non-diagonal elements, each equal to $\frac{1}{v}$. Let $M(v,u)=uM(n-1)+(1-u)M(v)$ where $0 \le u \le 1$.
We note that any $M \in M(v,u)$ is doubly stochastic with permutations (i.e., $M(1)$) and the uniform matrix (i.e., $M(n-1$)( as special cases.
We believe this family of matrices can capture some of the complexities of real network traffic in a more concise framework and is similar in spirit to the traffic models used in \cite{bojja2016costlyEclipse,solsticeliu2015scheduling,valls2021birkhoff}, and others.

\subsection{Demand Matrix and the Network Model}\label{subsec:DemandMatrixNetworkModel}
 Our network is modeled in a similar manner to the ToR-Matching-ToR (TMT) model, which was recently the basis for several reconfigurable datacenter networks (RDCN) architectures \cite{mellette2017rotornet,ballani2020sirius,griner2021cerberus,zerwas2023duo}. 
The network illustrated schematically in \autoref{fig:OrthNetModelIntro}, comprises two layers. A spine switch layer and a leaf switch layer. The leaf layer has $n$ nodes representing top-of-the-rack (ToR) switches $T_1,...,T_n$, and the spine layer has a \emph{single} reconfigurable (dynamic) switch.\footnote{The single switch  
is a special case of the TMT. However, the single switch abstraction was also used in previous works \cite{wang2009your,farrington2010helios,ballani2020sirius}, and we believe it could be generalized in future work.} 
Internally, each ToR in the leaf layer has a host that generates demand. The host is connected to the ToR in a single bidirectional connection. Furthermore, each ToR has one bidirectional port towards the spine switch. The spine switch has $n$ input and $n$ output ports of equal capacity, connected via a single directed \emph{matching} between its input and output ports, creating a directed topology. In our model, The spine switch can dynamically change its matching while the network operates, called a \emph{reconfiguration}. 
As mentioned before, at any time the switch implements a single matching. In this matching, any node could be connected to any other node via a single bidirectional connection but only to one other node.  We use the assumption that only single-hop routing is allowed during each matching. 
However, in our model, to facilitate two-hop routing, data sent to a node in a particular matching can be 'stored' in the ToR (via the host, for example, as in \cite{mellette2017rotornet}) and then forwarded to its destination in a future matching.
This assumption allows us to simulate networks of arbitrary degrees over time. 
We also assume that reconfiguration incurs a temporal cost in the form of \emph{reconfiguration time}, or delay, which we will denote as $R$. The \emph{transmission rate} of all the ports of the switches connected to or from the spine switch in our network is equal and denoted as $r$, and,w.l.o.g, will be typically set to $r=1$. We assume that the transmission rate of the host-to-ToR link is infinite. This last assumption allows us to focus on the topology formed between the ToRs. A similar assumption was used in \cite{jyothi2016measuring}.

Let $M$ denote our $n \times n$ demand matrix in bits of traffic between ToRs. That is $M[i,j]$ is the amount of traffic that ToR (node) $T_i$ needs to send to ToR $T_j$.
Unless otherwise stated, we assume that $M$ is doubly stochastic (i.e., $r=1$).

\subsection{The Topology and the Traffic Schedules}\label{subsec:ScheduleDefinition}

To fully define a dynamic network, we need to define its topology and how traffic is sent on each topology at each time slot.
To do so we define two types of schedules, a topology schedule \ts and a traffic schedule \ps.
We will use \autoref{fig:example} as an example, but we first formally define them.

\para{Topology schedule}  
A topology schedule \ts as an ordered set of $v=\card{\ts}$ tuples of three elements $\{\configElem_i,\alpha_i, R_i\}$. In each tuple, the first element is a permutation matrix $\configElem_i \in P_\pi$ (which can be translated to a switch configuration as a matching). 
The second element of the tuple is the transmission time $\alpha_i > 0 $, the time the system holds active the i$'th$ switch configuration, $\configElem_i$, 
which we also denote as \emph{time slot} $i$.
The third element is a reconfiguration time $R_i \ge 0$: the time it takes to reconfigure between consecutive switch configurations, i.e., replacing the $\configElem_i$ configuration with the $P_{i+1}$ configuration. No transmission is possible during this time.
Formally, the topology schedule is 
\begin{equation*}
    \ts=\{(\configElem_1,\alpha_1,R_1),(\configElem_2,\alpha_1,R_2),...,(\configElem_v,\alpha_v,,R_v)\}
\end{equation*}

We define the \emph{completion time} (denoted as \dctSchedule) of a topology schedule \ts as the sum of all transmission times $\alpha_i$ plus the reconfiguration times $R_i$. Formally, 
\begin{align}\label{eq:scheduleGeneral}
 \dctSchedule(\ts)&=  \sum_{i=1}^{\card{\ts}} (\alpha_i  +R_i) 
\end{align}

\begin{figure*}[ht]
    \centering
    \begin{tabular}{ccc}
       \raisebox{0.5cm}{\includegraphics[width=0.15\textwidth]{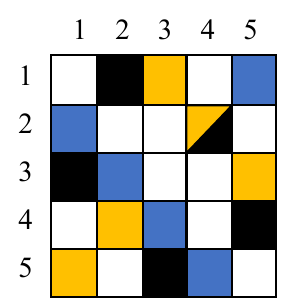}}  &
      \includegraphics[height=3.99cm]{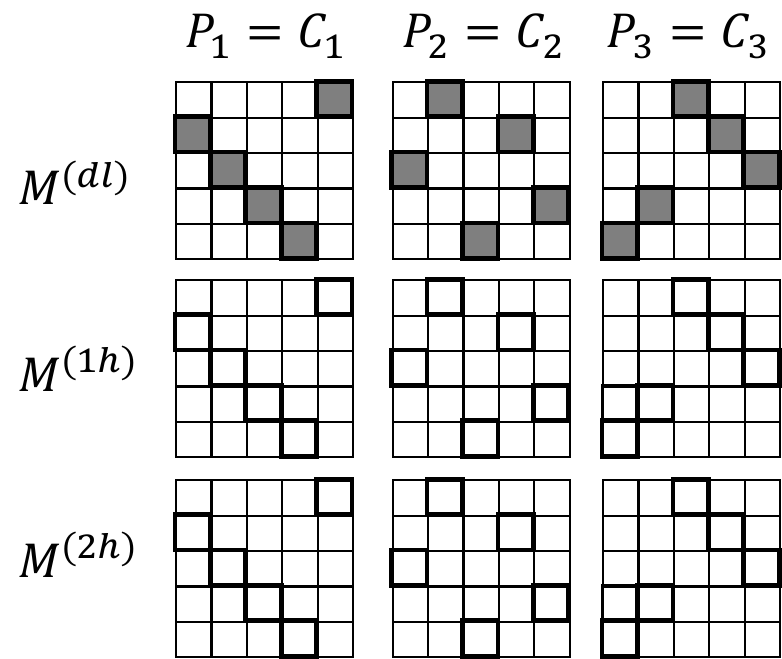} & 
      \includegraphics[height=3.99cm]{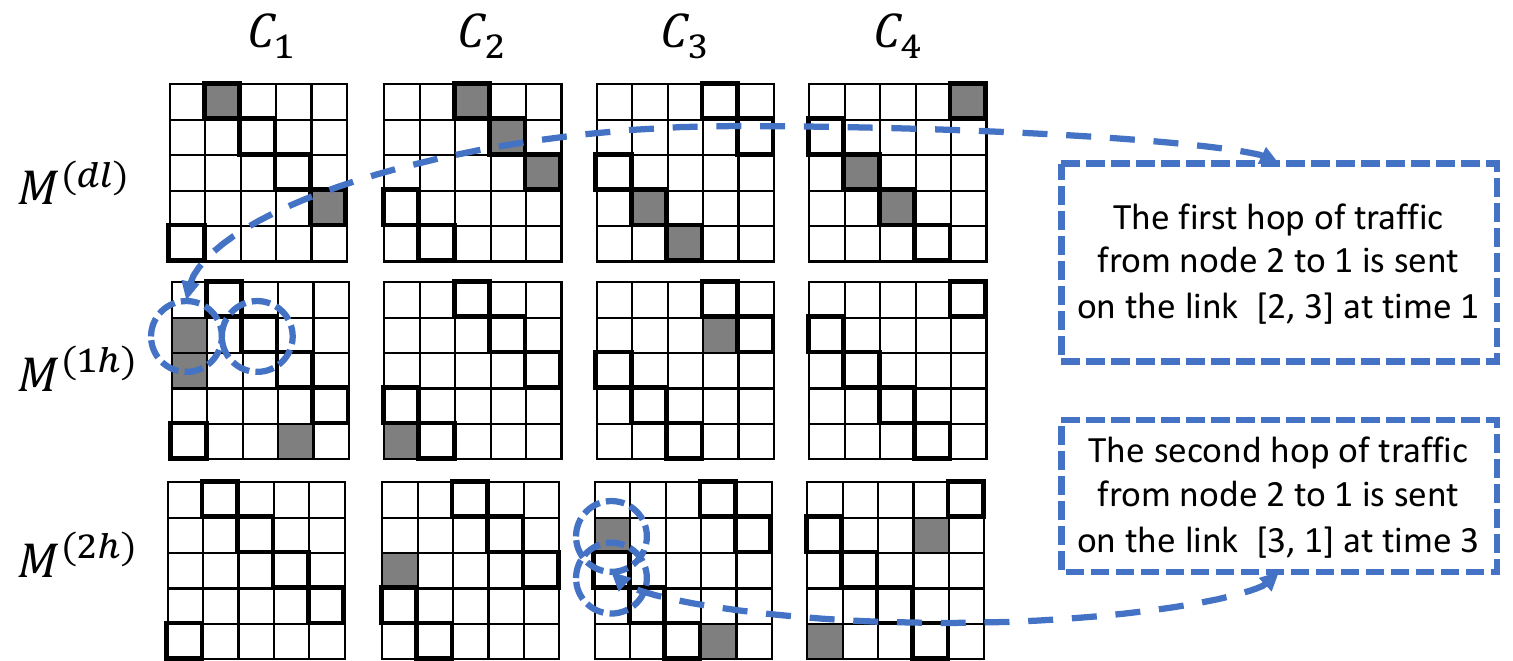} \\
      (a) Demand & \hspace{1cm} (b) BvN-based schedules & \hspace{-2.2cm} (c) Round-Robin-based schedules       
    \end{tabular}
    \caption[An example of a demand matrix and schedules that serve it]{An example of demand and schedules, (a) Double stochastic demand matrix $M$. The values in each cell is 0, $\frac{1}{3}$, or $\frac{2}{3}$ depend on the number of colors in the cell. (b) a simple \bvn-decomposition-based schedules which decompose $M$ into three permutation matrices $(P_1, P_2, P_3)$ (that become the topology schedule $(\configElem_1, \configElem_2, \configElem_3)$) and use only a single hop routing, i.e., $\Mdl$. (c) a round-robin-based topology schedule $(\configElem_1, \configElem_2, \configElem_3, \configElem_4)$ that is predetermined (and oblivious to the demand) and a traffic schedule that also uses 2-hops paths, $\Moh$, and $\Mth$.  
    }
    \label{fig:example}
\end{figure*}

\para{Traffic Schedule}
 A \emph{detailed traffic schedule} defines the traffic sent at each time slot $i$, i.e., during the switch configuration $\configElem_i$. 
The detailed traffic schedule $\ditPS$ is an ordered set of entries, which we assume are packets (though these could also represent traffic flows). Each entry $\sentry_i$ is a tuple with the following six elements: 
(the time slot at which the packet is sent, the packet size, the original source of the packet, the packet's final destination, the packet's current source, and the packet's current destination). We denote these elements as  $\sentry_i=\{t_i,w_i,s_i,d_i,x_i,y_i\}$.
Since the packet is sent over configuration $\configElem_i$, a valid entry must have $\configElem_i[x_i,y_i] = 1$ or alternatively, the edge $(x_i,y_i)$ is in the current matching $\configElem_i$.
In this work, we consider only paths of at most two hops.
For a single hop path $(s_i,d_i)$, we denote the packet sent as \emph{direct local}  \cite{mellette2017rotornet} and $\sentry_i$ will be of the form $\{t_i,w_i,s_i,d_i,s_i,d_i\}$.
For a two hops path 
$s_i \rightarrow q_i \rightarrow d_i $ we denote the packet sent in the first hop as \emph{indirect local} \cite{mellette2017rotornet} and $\sentry_i$ will be of the form $\{t_i,w_i,s_i,d_i,s_i,q_i\}$. The packet sent in the second hop is denoted as \emph{direct non-local} \cite{mellette2017rotornet} and $\sentry_i$ will be of the form $\{t_i,w_i,s_i,d_i,q_i,d_i\}$.
To conclude, a detailed traffic schedule defines all the traffic being sent over a topology schedule.
 
The schedule represents the \emph{amount} of traffic sent between each source and destination pair, and we use three sets of metrics to represent the traffic scheduling. This work considers only schedules that offer up to two hop paths, so we only need three different matrices to summarize the traffic at time slot $i$: A matrix, $\mdl_i$ for the traffic sent via single hop (i.e., direct local), a matrix, $\moh_i$ for the sent packets that are the first hop of a two hops paths traffic (i.e., local indirect), and a matrix, $\mth_i$, for the packets that are second hop of a two hops traffic (i.e., direct non-local) traffic. Also note that in our model, we allow a packet to traverse at most a single hop in each configuration $\configElem_i$. 
 
Next, we define these metrics formally.
Let $\mathds{1}_A$ be an indicator function for condition $A$.
Every cell of $\mdl_i$, represents direct local traffic sent from $s$ to $d$ at time slot $i$.
$\forall s,d$ s.t. $\configElem_i[s, d]=1$, 
\begin{align}
\small
\mdl_i[s, d]:=  
 \sum_{\sentry_j} 
 w_{j} \cdot \mathds{1}_{t_j = i \land s_j = x_j = s \land d_j = y_j = d}, 
\end{align} 
and zero otherwise.
Note that for this case if $\mdl_i[s,d] > 0$, only if  $\configElem_i[s,d]=1$.
Next, every cell of $\moh_i$ represents indirect local traffic, that is, traffic on its first hop out of two, at time slot $i$. For each two hops 
path $s \rightarrow q \rightarrow d$, it update the cell $\moh_i[s, d]$ when $\configElem_i[s,q]=1$. Formally,
$\forall s,d \text{ s.t. } \configElem_i[s,q]=1$,
\begin{align}
\small
 \moh_i[s, d]:=  
\sum_{\sentry_j} 
 w_{j} \cdot \mathds{1}_{t_j = i \land s_j = x_j = s \land d_j = d \land y_j = q}, 
\end{align}

and zero otherwise.
Note that for this case if $\configElem_i[s,d]=1$, then $\moh_i[s, d]$ is zero.
Every cell of $\mth_i$, represents direct non-local traffic.  That is traffic on its last hop of two hops in time slot $i$. For each two hops 
path $s \rightarrow q \rightarrow d$, it update the cell $\moh_i[s, d]$ when $\configElem_i[q, d]=1$. Formally, 
$\forall s,d \text{ s.t. } \configElem_i[q, d]=1$, 
\begin{align}
\small
 \mth_i[s, d]:=  
\sum_{\sentry_j} 
 w_{j} \cdot \mathds{1}_{t_j = i \land s_j = s \land x_j = q \land d_j = y_j = d}, 
\end{align}
and zero otherwise. Note that for this case if $\configElem_i[s,d] = 1$, then $\moh_i[s, d]$ is zero.
Formally, a traffic schedule, $\ps$, is then defined as
$$
    \ps = \{(\mdl_1,\moh_1,\mth_1), \dots ,(\mdl_v,\moh_v,\mth_v) \}.
$$

Let $\Mdl, \Moh$ and $\Mth$  be the ordered sets of the three types of matrices, $\Mdl = \{\mdl_1, \dots, \mdl_v\}, \Moh = \{\moh_1 \dots, \moh_v\}$ and $\Mth =\{\mth_1, \dots, \mth_v\}$. So, $\Mdl$ is the set of the entire direct local schedule, $\Moh$ is the set of the entire indirect local schedule, and $\Mth$ is the set of the entire direct non-local schedule.

Note that the traffic scheduling explicitly defines what packets are sent in each time slot $i$ and on each link $(k,l)$ s.t $\configElem_i[k,l]=1$, in particular, the \emph{total} amount of bits that are sent on the link $(k,l)$ at time $i$ is: $\hat{m}_i[k,l]=\mdl_i[k,l] + \sum_d  \moh_i[k,d] + \sum_s \mth_i[s,l]$. That is the sum of direct, first-hop, and second-hop traffic. 
We define $\totalDemandMat$ as the \emph{total traffic matrix}:
\begin{definition}[Total traffic matrix]\label{def:Totalmatrix}
Given a traffic schedule, $\ps$, we define the total traffic matrix $\totalDemandMat$ as a matrix where each cell $\totalDemandMat[k,l]$ denotes the \emph{total} amount of bits sent on the link  $(k,l)$ throughout the entire schedule:  $\totalDemandMat[k,l] =\sum_i \hat{m}_i[k,l]$. 
\end{definition}

We say that a pair of topology and traffic schedules $(\ts, \ps)$ is \emph{feasible} only if several conditions hold.
These conditions include, for example, the requirements that second-hop packets can be sent only after the first-hop was sent, every (indirect) packet that is sent on a first-hop needs to be sent on a second-hop, packets at time slot $i$ can only be sent on existing links in the configuration of time slot $i$, the time to send the amount of traffic in a time slot can be no longer than the time slot length, etc. We list these conditions formally in the appendix at \autoref{app:subsec:feasibltySchedule}, and for now, consider only feasible schedules.

In general, we would like our traffic schedules to send all of the demand; we define this formally as a \emph{complete schedule}, 
\begin{definition}[Complete schedules]\label{def:completeSchedule}
Let $M$ be a demand matrix and \ts and \ps a pair of feasible schedules. 
We say that \ts and \ps are \emph{complete} schedules for $M$ if all traffic in $M$ is sent and arrives at the destination.
\begin{align}
 \sum_{i=1}^v\mdl_i+\sum_{i=1}^v\moh_i=\sum_{i=1}^v\mdl_i+\sum_{i=1}^v \mth_i=M   
\end{align}
That is, all traffic in $M$ is sent and arrives at the destination.
\end{definition}

\autoref{fig:example} presents an example for the above concepts. In (a) we show a double stochastic demand matrix $M$,
the value in each cell is 0, $\frac{1}{3}$, or $\frac{2}{3}$ depending on the number of colors in the cell. The sum of each row and column is one.
In \autoref{fig:example} (b), we describe a simple \emph{topology schedule} $(\configElem_1, \configElem_2, \configElem_3)$, which is based on a \bvn matrix decomposition into three permutation matrices. Each $\configElem_i$ is highlighted by thick cells. For each $\configElem_i$ we have $\alpha_i=\frac{1}{3}$.
The \emph{traffic scheduling} uses only single-hop paths, so $\Moh$, and $\Mth$ have only zero matrices and $\mdl_i=\alpha_i \configElem_i$ for each time slot $i$. In this case, \ts and \ps are \emph{complete schedules} since $\sum_i \mdl_i = M$. 

In \autoref{fig:example} (c), we present what we later call a round-robin topology schedule. This schedule (i.e., $(\configElem_1, \configElem_2, \configElem_3, \configElem_4)$ shown in thick cells), is predetermined, oblivious to the demand and rotates between a set of permutations that creates a complete graph \cite{mellette2017rotornet}. Also here $\alpha_i=\frac{1}{3}$. The traffic schedule for this case uses 2-hops paths, as seen in $\Moh$, and $\Mth$. For example, at time slot 1, node 2 sends a packet to node 3 (toward node 1) as shown in $\moh_1[2,1]$, and in time slot 3, node 3 forwards this packet to node 1, as shown in $\mth_3[2,1]$.
In this case, \ts and \ps are also \emph{complete schedules} since $\sum_i \mdl_i + \sum_i \moh_i = M$.

 
\subsection{Problem Definition: The DCT and Throughput of a System }\label{subsec:problemdefinition}
 This section defines the main metric of interest studied in this work.
 In the spirit of previous works, we are interested in the \emph{throughput} of different network designs, or as we denote them here, \emph{systems}. We use the \emph{demand completion time} (DCT) explained below as a proxy to study the throughput.
While previous analytical works studied separately topology and traffic, i.e., either systems with given static topologies \cite{namyar2021throughput}, or dynamic topologies with naive traffic scheduling \cite{bojja2016costlyEclipse,valls2021birkhoff}, or only one specific kind of system \cite{addanki2023mars}, we extend the model to support the integration of \emph{both} topology and traffic schedules.
 To the best of our knowledge, this is the first work that formalizes the joint optimization of both the topology and the traffic scheduling, enabling us to compare different systems: dynamic, static, demand-aware, or demand-oblivious, and with different parameters, on a simple model. 
 
 We take the worst-case approach of a classical \emph{throughput} definition by Jyothi et al.~\cite{jyothi2016measuring} and several followups  \cite{griner2021cerberus,namyar2021throughput,addanki2023mars} and define the \emph{demand completion time} (and accordingly the throughput) of a system, as a function of the \emph{worst-case} demand matrix for that system. 
 Following this definition, and considering, for example, all double stochastic matrices, we can say that a system that has a \emph{system DCT} of $x$ seconds can serve \emph{any} demand matrix in less or equal $x$ seconds. In turn, we can compare different systems by their system DCT (or throughput).

 A system $\sys$ is defined as the set of all possible pairs of schedules $(\ts, \ps)$ that it supports.
 To give some examples, a system  
 could be defined by all schedules with a specific reconfiguration time resulting from the system hardware, or all schedules that use only single-hop routing, all schedules that have a demand-oblivious topology, or all schedules that use only a static topology. 
 In the next section, we will formally define several such known systems. 

 A \emph{topology scheduling algorithm}, $\stalg$, in an algorithm that generates a topology schedule, and a \emph{traffic scheduling algorithm}, $\spalg$ in an algorithm that generates a traffic scheduling. We say that a pair of algorithms belong to a system $\sys$ when their generated schedules (for any demand) are in  $\sys$, i.e.,  $(\stalg, \spalg) \in \sys$.
 We can now define the different types of completion times, denoted as $\dct$,
 starting with the \emph{matrix completion time}.

\begin{definition}[Matrix completion time]
For a given demand matrix $M$ and a pair of schedule algorithms $(\stalg, \spalg)$ that generate two feasible and complete schedules $(\ts, \ps)$, the completion time of the matrix $M$ is defined as the time it takes to execute the schedules, 
 \begin{align*}
     ~\dct(\stalg(M), \spalg(M), M)&= \dctSchedule(\stalg(M)).
  \end{align*}
\end{definition}

Going back to the example of \autoref{fig:example}, in (b), the 
matrix completion time is $3 \cdot \frac{1}{3} + 3 R_b$ where $R_b$ is the reconfiguration time of the schedule in (b). In (c), the 
matrix completion time is $4 \cdot \frac{1}{3} + 4 R_c$ where $R_c$ is the reconfiguration time of the schedule in (c). Therefore, to determine which of the schedules is faster, we must know the relations between $R_b$ and $R_c$, e.g., if they are equal, setup (b) is faster.  

The demand completion time for \emph{a pair of algorithms} is defined as the matrix completion time of their \emph{worst-case} matrix. 

\begin{definition}[Algorithms completion time]
Let $\mathcal{M}$ by the family of all double stochastic matrices. 
Let $(\stalg, \spalg)$ be a pair of schedule algorithms.  
The algorithms' completion time is defined, 
 \begin{align*}
      \dct(\stalg, \spalg, \mathcal{M})&= \max_{M \in \mathcal{M}} \dct(\stalg(M), \spalg(M), M).
  \end{align*}
\end{definition}

We can now finally define the \emph{demand completion time} of a system $\sys$, denoted as \emph{system DCT},

\begin{definition}[System DCT]\label{def:sysDCT}
Let $\mathcal{M}$ by the family of all double stochastic matrices. 
Consider a network design system $\sys$.
The demand completion time of a $\sys$ is defined as the completion time of its best algorithms,  
 \begin{align*}
      ~\dct(\sys, \mathcal{M})&= \min_{(\stalg, \spalg) \in \sys} \dct(\stalg, \spalg, \mathcal{M}).
  \end{align*}
\end{definition}

As noted in previous work \cite{griner2021cerberus}, the DCT has a direct relation to the \emph{throughput} of the system. In particular, it follows that the throughput of a system $\sys$, defined in bits per second, is the reciprocal of the system DCT, $\dct(\sys)$. For complete schedules and double stochastic matrices, we then define 
\begin{align}
    \mathrm{Throughput}(\sys, \mathcal{M}) = \frac{1}{\dct(\sys, \mathcal{M})}.
\end{align}

In the rest of the paper, we are set to study the system DCT of different systems.

We define three major systems, \csn, \rsn, and \msn, and explore their properties and tradeoffs in the journey to better understand the DCT and throughput of different systems. We show that our model from \autoref{sec:modelAndProbDefinition} fits known systems like \csn and \rsn, enabling us to study their DCT and derive familiar and new results, in particular, we use these results in our later analysis on \msn.
Recall that our demand is represented by an $n\times n$ saturated doubly stochastic matrix $M$, which means that $\Wcard{M} = nr$ and $r=1$. 
However, when appropriate, we write these parameters explicitly so the bounds and results are generalized to \emph{scaled} doubly stochastic matrices that are not necessarily saturated, i.e., $\Wcard{M} < nr$.
In our analysis of the system DCT, we assume that the delay is only due to switch reconfiguration 
and transmission times. 
We neglect other factors, such as packet loss and congestion 
and focuses on the basic tradeoffs.

\section{The Demand-Aware Topology System: {\csn}}
\label{sec:NetwrokDesignsCSN}

We first consider in our formal model the demand completion times of the well-established system, \csn.
The \csn system is \emph{demand-aware}, meaning that the topology schedule is computed based on a decomposition of the demand matrix $M$.  The traffic schedule of the \csn, in turn, is very simple and supports only single-hop paths.
Examples of networks with this type of system are FireFly \cite{hamedazimi2014firefly}, Eclipse \cite{bojja2016costlyEclipse}, Solstice \cite{solsticeliu2015scheduling}, and  ProjecTor \cite{ghobadi2016projector}.

Recall that the topology schedule \ts is a list of matchings $\configElem_i$ and coefficients $\alpha_i$. In \csn we assume that \ts is generated through a \bvn~decomposition, and for each $i$ in the decomposition we simply set the switch configuration $\configElem_i=P_i$.
The specific \bvn algorithm can vary, but the goal is to minimize the number of reconfigurations (in the simulation, we use the state-of-the-art algorithm \cite{valls2021birkhoff}). 
Since $\ts$ is based on a \bvn~decomposition and $M$ is doubly stochastic the following will hold, $\sum_{i=1}^{\card{\ts}} \beta_i P_i = M$ and the sum of coefficients is one, that is, $\sum_{i=1} \beta_i=1$. 
The transmission time is then $\alpha_i = \beta_i/r$.  
In \csn and its topology schedules, we assume that the reconfiguration time between matchings is a constant, denoted $\crec$, and
is identical to all times, i.e., $\forall i, \;\; R_i = \crec$.
The reconfiguration time for \bvn systems is considered \emph{high}, in particular in comparison to demand-oblivious topology schedules. This is assumed for two reasons: i) either the flexibility that the topology requires (enabling any possible matching) results in a larger time to set the configuration and/or ii) the time to \emph{compute} the next matching (usually in some greedy way) requires much more time than in a precalculated matching. 

The traffic scheduling of \csn, $\ps$, supports only single-hop paths, which means
each element in the direct local matrix equals $\mdl_i[k,l]=\beta_i$ where $\configElem_i[k,l]=1$ and zero otherwise.
In other words, the traffic schedule is the same as the topology schedule in terms of active matrix cells. 
Since we only use direct local traffic, $\Moh$ and $\Mth$ are sets of zero matrices. Recall that the example of \autoref{fig:example} (b) showed topology and traffic schedules of \csn.

In \csn, both the topology and the traffic schedules are of specific types. Let $\stalg=\bvn$ be an algorithm that decomposes $M$ into $v=\card{\ts}$ matching (i.e., permutation matrices) and generates the topology schedule.
The algorithm that generates the traffic schedule that sends only direct local traffic is denoted as $\spalg=\direct$.
We can define \csn formally in the following manner.
\begin{definition}[The \bvn-System, \csn ]\label{def:main:csnSysDef}
       \csn is defined as a system that must use a topology schedule $\bvn$. 
\end{definition}

Following Eq. \eqref{eq:scheduleGeneral} we get that for a  double stochastic matrix $M$ and $r=1$ the \emph{matrix} completion time is
\begin{align}\label{eq:damatDCT}
    &\dctda(\bvn,\direct,M)=  \sum_{i=1}^v (\alpha_i  +\crec) =1 +v \crec 
\end{align}
We note that the matrix DCT is linear in the number of reconfigurations $v$.
It will be minimal when the number of reconfigurations is one, namely a single 
permutation matrix. 
Regarding the worst-case matrix, the \bvn decomposition of any \emph{positive} demand matrix (i.e., where all elements other the diagonal are positive) requires at least $n-1$ permutations and so $n-1$ reconfigurations to be fully transmitted, one such example is the \emph{uniform} matrix. We can therefore state the following observation on the system DCT of \csn.
\begin{observation}[The system DCT of \csn ]\label{obs:main:DAsysDCTBound}
The system DCT of \csn for $r=1$ is at least
    \begin{align}\label{eq:dasysDCT}
    \dctda(\csn, \mathcal{M}) &\ge 1 +(n-1) \crec
\end{align}
\end{observation}

\section{The Demand-Oblivious System: {\rsn}}
\label{sec:NetwrokDesignsRSN} 
This system is based on a \emph{demand-oblivious} topology
and is denoted as the round-robin systems, \rsn. 
The \rsn system is motivated by recent proposals for datacenters network design like \rotornet \cite{mellette2017rotornet}, Opera \cite{mellette2020expandingOpera}, and Sirius \cite{ballani2020sirius}.
These types of networks are also denoted as 'rotor'-based designs.

In \rsn the system topology is composed of a pre-determined and pre-computed set of $n-1$ matchings (i.e., permutation matrices, switch configurations) 
whose union forms a \emph{complete graph}. 
Let $\cycleSched=\{\cycleSchedElem^1 ,...,\cycleSchedElem^{n-1 } \}$ be the set of switch configurations denoted as \emph{cycle configuration}, where each $\cycleSchedElem^i$ is a permutation matrix.

In \rsn, the topology scheduler $\ts$ transitions from one switch configuration in $\cycleSched$ to another in a periodic, round-robin and fixed manner, 
oblivious to the demand. The period of $n-1$ switch configurations that emulate the complete graph is called a \emph{cycle}. 
Each configuration $\cycleSchedElem^i$ is held for a time $\delta$, i.e., $\forall i, \alpha_i=\delta$.
The reconfiguration time in the topology schedule is also constant and denoted as $\rrec$, $\forall i, R_i=\rrec$. An important property of the \rsn topology schedule is the \emph{duty cycle} as defined in prior work \cite{mellette2017rotornet,mellette2020expandingOpera} to be $\eff$ where $\eff= \frac{\delta}{ \delta+\rrec}$. This is the fraction of time spent sending data, $\delta$, to the total time spent in a slot, $\delta +\rrec$. 
We denote the algorithm that generates such a round-robin topology schedule as $\stalg=\rr$.
We define \rsn formally in the following manner.
\begin{definition}[The Round-Robin System, \rsn]\label{def:main:rsnSysDef}
    \rsn is defined as a system that must use $\rr$ as its topology schedule. 
\end{definition}

We discuss next the set of traffic schedulers that \rsn is may choose from.
Up-to-date, there is no single, optimal, well-accepted traffic scheduling algorithm for \rsn.
Previous work \cite{ballani2020sirius} considered 2-hops Valiant routing \cite{valiant1981universal}. While this approach uses oblivious routing and is easy to implement, it is not optimal regarding worst-case demand completion times for \rsn. A different example, \rotornet, uses RotorLB  \cite{mellette2017rotornet}, a decentralized protocol that uses adaptive routing and allows nodes to negotiate routing decisions. Briefly, in RotorLB, single hop is preferred over two hops, and  nodes only help other nodes if they have available capacity, thus spreading the load more evenly. The bounds of this approach are not clear.
Nevertheless, here we present several algorithms and bounds that enable us to find the system DCT of \rsn formally.

\subsection{The system DCT of \rsn}
%

 
Recall that in our model, traffic can be sent \emph{directly} (in a single hop) \emph{or} \emph{indirectly} (using two hops). A traffic scheduling algorithm for \rsn should decide which traffic in the demand matrix $M$ is routed directly using a single hop and which traffic is routed indirectly in two hops on the \rsn topology schedule $\rr$.
Our first result is a lower bound for the system DCT of \rsn,
by a lower bound on the completion time of a permutation demand matrix for any traffic scheduling algorithm.

\begin{restatable}{theorem}{LowerBoundPerm}[Lower bound for \rsn system DCT]\label{thm:LowerBoundPerm}
Let $\spalg$ be any traffic scheduling algorithm for \rsn that generates a complete schedule. Let $P \in P_{\pi}$ be any derangement permutation matrix, then we have   
\begin{align}
 \dctrot(\rr, \spalg, P) &\geq (2-\frac{2}{n})\frac{1}{ \eff r} \frac{ \Wcard{P}}{n}
\end{align}
\end{restatable}
We provide details of the proof in the appendix in \autoref{sec:LowerBoundPerm}. However, for brevity, we provide the following proof overview.

\para{Proof overview for \autoref{thm:LowerBoundPerm}}
Our proof has two main parts. We first prove that the lower bound is a function of the skew parameter we denote as $\skew$, then prove an upper bound on this parameter.
The skew parameter, $\skew$, denotes the fraction of traffic sent via a single hop in a given traffic schedule, and $1-\skew$ is the fraction that is sent on two hops.
It follows that $(2-\skew)$ is the average path length of a given schedule.
We give a formal definition for $\skew$ in the appendix using \autoref{def:skew}.  
Following our definition,  In \autoref{thm:lowerBoundRSN}, we show that the DCT of \rsn is lower bounded by $(2-\phi)\frac{1}{ \eff r} \frac{ \Wcard{M}}{n}$. 

To find a lower bound on the DCT, we need to increase $\skew$ (and to decrease the average path length). We, therefore, in \autoref{thm:lowRSNApp}, prove the upper bound for $\skew$ with any optimal traffic scheduler \spalg on any demand permutation matrix in $M(v)$.
As part of the proof of our theorem, we state the following corollary of \autoref{thm:lowerBoundRSN} and \autoref{thm:lowRSNApp}.

\begin{corollary}[Lower bound for $M(v)$ on \rsn]\label{cor:lowerBoundRSNMV}
Let $\spalg$ be any traffic scheduling algorithm for \rsn that generates a complete schedule. Given an $n\times n$ matrix $M\in M(v)$, the \emph{lower bound} for the DCT of \rsn is
   \begin{align}
    \dctrot(\rr,\spalg, M(v))\geq (2-\frac{2v}{n-1+v})\frac{1}{ \eff r} \frac{ \Wcard{M}}{n}
\end{align} 
\end{corollary}
This results from combining the lower bound on the DCT of \rsn with the upper bound on $\skew$ for matrices from $M(v)$.
Finally, we note that the worst case if for $M(1)$, i.e., $v=1$, that represents the case of permutation matrix $P\in P_\pi$, and imply $\skew=\frac{2}{n}$ which proves Theorem \ref{thm:LowerBoundPerm}. 

Next, we consider a traffic scheduler, which, surprisingly, achieves the lower bound we stated before for any derangement permutation matrix $P\in P_\pi$ on \rsn. 
We denote this as the permutation scheduler \algOnePerm.
Briefly, \algOnePerm, generates a \emph{detailed scheduled} which equally distributes the demand from the permutation demand cells such that each connection will eventually send a total of $\frac{2}{n}$ (normalized) bits. We provide a full description of this algorithm in the appendix in \autoref{app:sec:AlgOnePerm}. This includes its pseudo-code in \autoref{alg:onePremSched} and a proof for the following lemma.

\begin{restatable}{lemma}{feasibleApp}[DCT for the \algOnePerm Algorithm]\label{lemma:feasibleApp}
Let $P\in P_\pi$ be an $n \times n$ scaled permutation matrix. 
Then, the schedule generated from the algorithm $\spalg=\algOnePerm$, is feasible, complete, and has a demand completion time, $\dctrot$ of,
\begin{align}
 \dctrot(\rr,\algOnePerm,P)= (2-\frac{2}{n})\frac{1}{ \eff r} \frac{ \Wcard{P}}{ n}
\end{align}
\end{restatable}

However, \algOnePerm only works for a single permutation. To find the system DCT, we need an algorithm for a general double stochastic matrix.
We next consider an algorithm that builds a schedule using \emph{multiple permutations}, we denote this scheduler as $\rrbvn$. This algorithm works in the following way: It decomposes any double stochastic demand matrix $M$ into a set of \emph{permutation matrices} (regardless of how many) and sends each using the permutation scheduler, \algOnePerm. The pseudo-code for this algorithm appears in \autoref{alg:rrbvn}. 
 Formally, we prove the following result about the DCT in the \rsn system, which matches the lower bound of \autoref{thm:LowerBoundPerm}.  

\begin{restatable}{theorem}{upperRSNBVN}[Upper bound for \rsn system DCT]
\label{thm:upperRSNBVN}
Let $M$ be a scaled doubly stochastic matrix. Then, using the traffic scheduling algorithm, $\rrbvn$ the matrix demand completion time in \rsn is 
\begin{align}
 \dctrot(\rr,\rrbvn, M) &= (2-\frac{2}{n})\frac{1}{ \eff r} \frac{\Wcard{M}}{n}.
\end{align}
\end{restatable}

\begin{proof}
Let $M$ be a $\lambda$ scaled doubly stochastic matrix. Using \bvn, we can decompose a scaled double stochastic matrix to $v$ different matchings. 
We get a list of $v$ coefficients that sum to $\lambda$, that is, $\sum_{i=1}^v\beta_i=\lambda$, and a list of $v$ permutation matrices $P_i$. 
We can then schedule each matching separately on \rsn using the algorithm \algOnePerm, as shown in \autoref{lemma:feasibleApp}.~
Recall that the sum of the product of the coefficients $\beta_i$ and permutation matrices $P_i$ is equal to the original traffic matrix, i.e.,  $\sum_{i=1}^v\beta_iP_i=M$. It also holds for the sum of the matrix weights $\sum_{i=1}^v\Wcard{\beta_iP_i}=\Wcard{M}$.
The DCT of \rrbvn is the sum of the DCTs of all the different permutations, that is
\begin{align*}
  \dctrot(\rr,\rrbvn,M)&=
  \sum_{i=1}^v \dctrot(\rr,\algOnePerm,\beta_iP_i) \\
  &=\sum_{i=1}^v  (2-\frac{2}{n})\frac{1}{ \eff r} \frac{ \Wcard{\beta_i P_i}}{n} \\
  &=   (2-\frac{2}{n})\frac{1}{ \eff r n}  \sum_{i=1}^v  \Wcard{\beta_i P_i}\\
  &= (2-\frac{2}{n})\frac{1}{ \eff r} \frac{ \Wcard{M}}{n}
\end{align*}
We would note that the length of each schedule changes. We, therefore, need to change the slot hold time, $\delta$, accordingly for each sub-scheduled permutation matrix so that the schedule finishes at the correct time.
Since the sum of all schedules achieves the value stated in the theorem, BvN will work for any doubly stochastic matrix the theorem follows. 
\end{proof}
We note that our proof and \rrbvn algorithm assumes that \rsn can use a dynamic value for $\delta$. This is very useful for the purpose of our proof and bounds; however, it might not be practically possible. Another option is to use a very small value for $\delta$, (which is indeed the case for related work \cite{mellette2017rotornet,mellette2020expandingOpera}). This would mean that \rsn might spend a short time idle and not transmitting at the end of the schedule, we leave the details of the practical implementation for future work.

While it was previously shown that permutation matrices are the worst case for \rsn type systems, both formally \cite{griner2021cerberus,Vamsiaddanki2023mars} and informally \cite{mellette2017rotornet,ballani2020sirius}, the above theorem closes a small gap and enables us to state the system DCT of \rsn formally. To summarize our \rsn analysis,  we offer the following theorem for the system DCT of \rsn, which follows directly from \autoref{thm:upperRSNBVN}.

\begin{theorem}[The system DCT of \rsn]\label{thm:main:RRsysDCT} Let $\mathcal{M}$ be the set of all double stochastic matrices. For $r=1$, and $\rrec=0$, the system DCT of \rsn is: 
   \begin{align}\label{eq:rrsysDCT}
    \dct(\rsn, \mathcal{M}) & =  2-\frac{2}{n}. 
\end{align} 
\end{theorem}

\begin{proof}
\autoref{thm:upperRSNBVN} proves that the bound stated in the theorem is feasible using the $\rrbvn$ algorithm for \emph{any} demand matrix. \autoref{thm:lowerBoundRSN} proves that
there is a matrix that takes at least the bound in the theorem.
This proves that the above bound is the best-case DCT for worst-case demand. 
\end{proof}

While we have shown an algorithm that achieves the system DCT for any double stochastic matrix, we can still improve the DCT of \rsn for some particular matrices. 
To see how, we first define a simple traffic schedule algorithm, $\stalg=\direct$. This algorithm creates a schedule by sending packets exclusively using a single hop. 
The following DCT bound was shown in previous work \cite{griner2021cerberus}, we repeat it here. We further give a proof in the appendix in \autoref{sec:directBoundProof}) to demonstrate that it is aligned with our definitions and model.

\begin{restatable}{theorem}{defRotDCTGenApp}[$\direct$ upper bound]\label{thm:defRotDCTGenApp}
Let $M$ be any demand matrix and $\direct$ be a traffic schedule where all packets are sent over a single hop. The matrix completion time in \rsn is then
 \begin{align}
    \dctrot(\rr ,\direct,M ) = (n-1)\frac{\max(M_{i,j})}{\eff r}
\end{align}    
\end{restatable}

For the worst-case demand matrix, the permutation matrix, this algorithm has a much worse DCT than $\rrbvn$. However, it will improve the DCT significantly when faced, for example with the uniform matrix.

The question of the best traffic schedule for \rsn on \emph{any} particular matrix is still open. Therefore, in the rest of the work, we consider a simple algorithm for \rsn that combines the best results between \autoref{thm:defRotDCTGenApp} and \autoref{thm:upperRSNBVN}.
Since both $\rrbvn$ and the naive $\direct$ traffic scheduling algorithms are defined on the general doubly stochastic matrix, we can combine them to form an improved algorithm $\upper$. This algorithm takes the best traffic schedule between the two, depending on the traffic matrix $M$ that is, 

\begin{definition}[$\upper$ algorithm  for \rsn DCT]\label{def:rsnUpperBound} 
\begin{align*}
    & \dctrot(\rr,\upper,M) =\\&= \min( \dctrot(\rr,\rrbvn,M), \dctrot(\rr,\direct,M)) 
\end{align*}
    
\end{definition}
 
Using this definition and comparing the expressions for DCT from \autoref{thm:upperRSNBVN} and \autoref{thm:defRotDCTGenApp}, we can arrive at the following observation. When the maximal element in $M$ is larger than nearly twice (i.e., by $(2-\frac{2}{n}$) the average matrix value $\frac{ \Wcard{M}}{n(n-1)}$ we will use the bound from \autoref{thm:upperRSNBVN} and otherwise, the bound from \autoref{thm:defRotDCTGenApp},
\begin{gather}
    \dctrot(\rr,\upper,M)=\nonumber \\=
\begin{cases}
  (2-\frac{2}{n})\frac{1}{ \eff r} \frac{ \Wcard{M}}{ n}, & 
    (2-\frac{2}{n})\frac{ \Wcard{M}}{ (n-1)n} \le \max(M_{i,j})\\  
\frac{(n-1)\max(M_{i,j})}{ \eff r},     &\text{otherwise}
\end{cases}
\label{eq:RRUPPER}
\end{gather}
We see that, indeed, $\upper$ cannot improve the result of $\rrbvn$ in the worst case.

In the next section, we consider a new alternative to the two existing systems we discussed, \msn, and study its performance.

\section{The Composite System and The Pivot Algorithm}\label{sec:msn}
Before introducing \msn, we make a simple observation following the previous sections that will serve as our main motivation.
The  \emph{worst-case} demand matrix for \rsn (i.e., a permutation matrix) is the \emph{best-case} demand matrix for \csn. The opposite is also almost the same: the \emph{best-case} demand matrix for \rsn (i.e., a uniform matrix) is a very bad case for \csn (but not the worst).
This motivates \msn, which is a system that can adapt to different types of demand using advantages from both the \csn and \rsn systems combined.
While both previously discussed systems, \rsn and \csn, were able to use a \emph{single} type of schedulers
we define \msn as a system and a network that can use both demand-aware and demand-oblivious types of topology schedules for the \emph{same} demand matrix.
Additionally, it must adjust its traffic schedule to make it feasible and complete.

One way to implement a schedule that uses demand-aware and demand-oblivious topologies is to decompose the demand matrix 
into two parts and use different schedule pairs for each part. We will define such an algorithm for \msn in this section.
But first, formally, we define \msn as the following system.

\begin{definition}[The Composite System, \msn]\label{def:main:msnSysDef}
    Let $M$ be any double stochastic demand matrix. \msn is defined as a system that decomposes $M$,  into two double stochastic matrices $M^{\bvn}$ and $M^{\rr}$, such that $M=M^{\bvn}+M^{\rr}$. It then schedules the topology and traffic for $M^{\bvn}$ with $\csn$ and sequentially  for $M^{\rr}$ using $\rsn$.

\end{definition}

\begin{algorithm}[t]
  \caption[The \algPivot\ Algorithm]{The \algPivot\ Algorithm}\label{alg:piovt}
  \begin{algorithmic}[1]
  \Require$M \in {\mathbf{R}^+}^{n\times n}$ scaled double stochastic
  \Ensure $M^{\bvn}$, $M^{\rr}$
   \State  $P_M=\{P_1,...,P_v\}$,$\beta_M=\{\beta_1,...,\beta_v\}$  \Comment{{\color{blue} Decompose $M$ to a set permutation matrices and coefficients using \bvn, sorted by $\beta_i$ from large to small}} 
   \State $D_0=\dctrot(\rr,\upper,M),D_v=\dctda(\bvn,\direct,M)$\label{alg:pivot:extreem}
\For {$i \in [1,v-1]$}
    \State    $M^{\bvn}=\sum_{j=1}^i\beta_j P _j$
    \State    $M^{\rr}=\sum_{j=i+1}^v\beta_j P _j$
    \State $D_i=\dctda(\bvn,\direct,M^{\bvn})$$+\dctrot(\rr,\upper,M^{\rr})$
\EndFor
  \State $f= \underset{i}{\argmin} ~D_i$  \Comment{{\color{blue} Find best decomposition}}
  \State  $M^{\bvn}=\sum_{j=1}^f \beta_j P _j$ and $M^{\rr}=\sum_{j=f+1}^v \beta_j P _j$\label{alg:pivot:dct}
    \end{algorithmic}
\end{algorithm}
Next, we describe a novel decomposition algorithm that we will use for our evaluation of \msn, both for the empirical and theoretical analysis.
Algorithm \algPivot (see \autoref{alg:piovt} for the pseudo-code), starts with  
\bvn matrix decomposition of the demand matrix $M$.
W.l.o.g., we assume that all permutations $P_i$ are sorted by their weight, that is, by their coefficients $\beta_i$, where $\beta_i$ is the largest coefficient.
Essentially, \algPivot is searching for the best partition of the permutations. It creates $v+1$ decompositions of $M$ to $M^{\bvn}$ and $M^{\rr}$. Each decomposition is defined by the index of a \emph{pivot permutation} where all permutations strictly heavier than the pivot permutation create $M^{\bvn}$ and all smaller create $M^{\rr}$. By comparing the expected DCT of each decomposition (Line \ref{alg:pivot:dct}), \algPivot returns the decomposition that minimizes the DCT. We note that \algPivot also examine the cases where $M^{\bvn}=M$ and $M^{\rr}=M$ (Line \ref{alg:pivot:extreem}). 

Let $\stalgmsn$ and  $\spalgmsn$ denote the topology and traffic schedules resulting for the \algPivot algorithm, respectively. Formally, $\stalgmsn$ is the concatenation of the topology schedules $\bvn(M^{\bvn})$ and $\rr(M^{\rr})$,
and $\spalgmsn$ is the concatenation of the traffic schedules $\direct(M^{\bvn})$ and $\direct(M^{\rr})$.

We can state the following general theorem for \msn and \algPivot
 \begin{theorem}[DCT of \msn with \algPivot]\label{thm:naiveBoundCompositeDCTWpivot}  
For any scaled doubly stochastic matrix $M$, the DCT of \msn is 
at most as the lowest DCT between \rsn or \csn, 
 \begin{align}
    & \dctmix(\stalgmsn,\spalgmsn,M )\leq \nonumber \\
    &\leq\min( \dctda(\bvn,\direct,M ), \dctrot(\rr, \upper,M))
\end{align}  
\end{theorem}
\begin{proof}
    This follows since \algPivot\ tests and compares the cases $(M^{\bvn}=M$, $M^{\rr}=0_n)$, and  $(M^{\bvn}=0_n$, $M^{\rr}=M)$.  Therefore, $\algPivot$ will be at least as good as either. 
\end{proof}

We discuss the running time of \algPivot~in the appendix \autoref{sec:sruuningtime}. But interestingly, we emphasize that \algPivot only needs to make a single \bvn decomposition during its execution.

As for the previous systems, we would like to find the system DCT for \msn, namely, to find the worst-case demand matrix for it and its completion time.

In the next section, we make a step in this direction, and we study the DCT of \msn and \algPivot and compare it to \csn and \rsn for the special case of the $M(v,u)$ family of matrices.

\begin{figure*}[t]
  \centering
  \begin{tabular}{cc}
  \multicolumn{2}{c}{\includegraphics[width=1\linewidth]{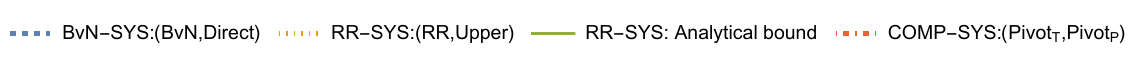}} \\
  \includegraphics[width=.39\linewidth]{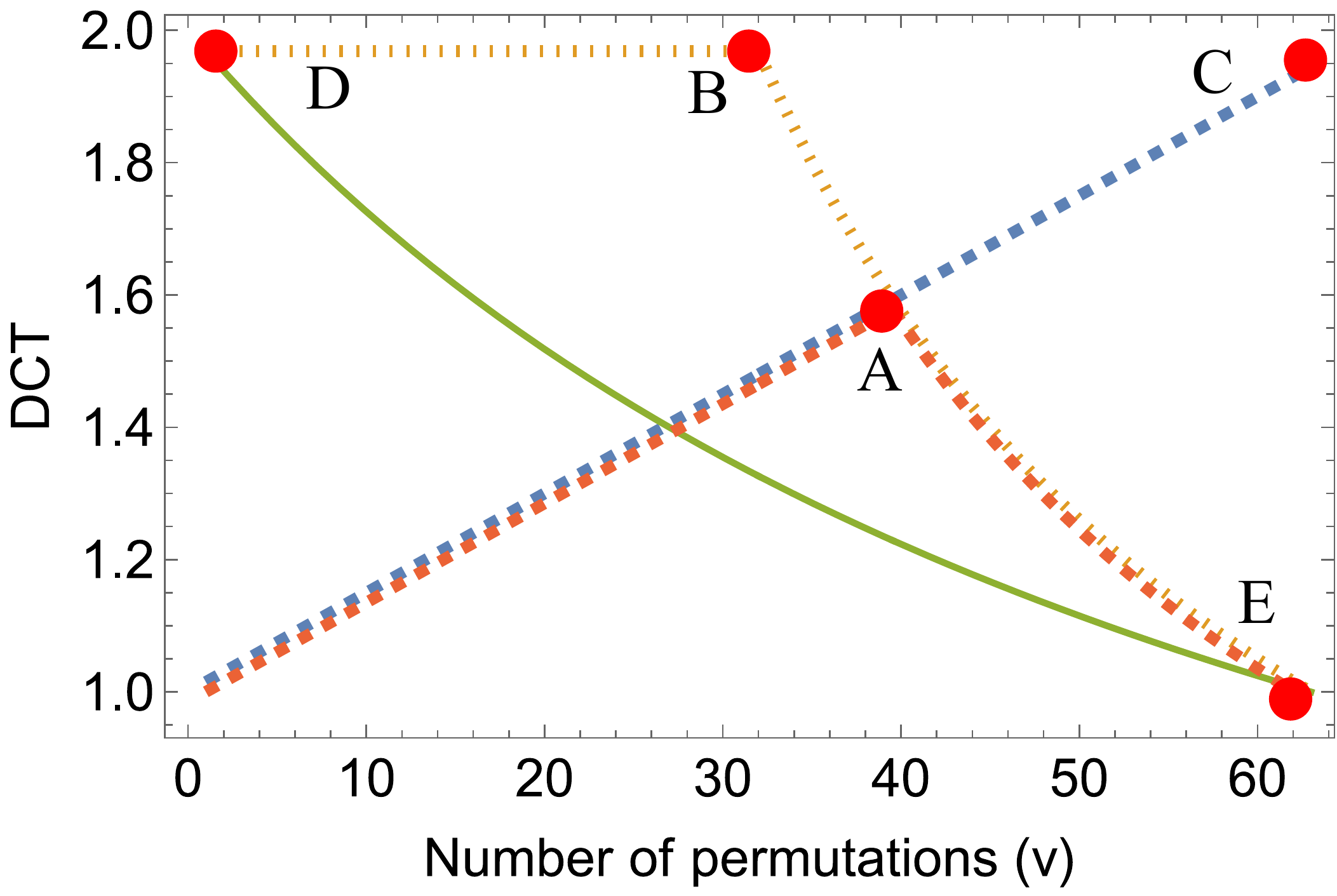} &
  \includegraphics[width=.40\linewidth]{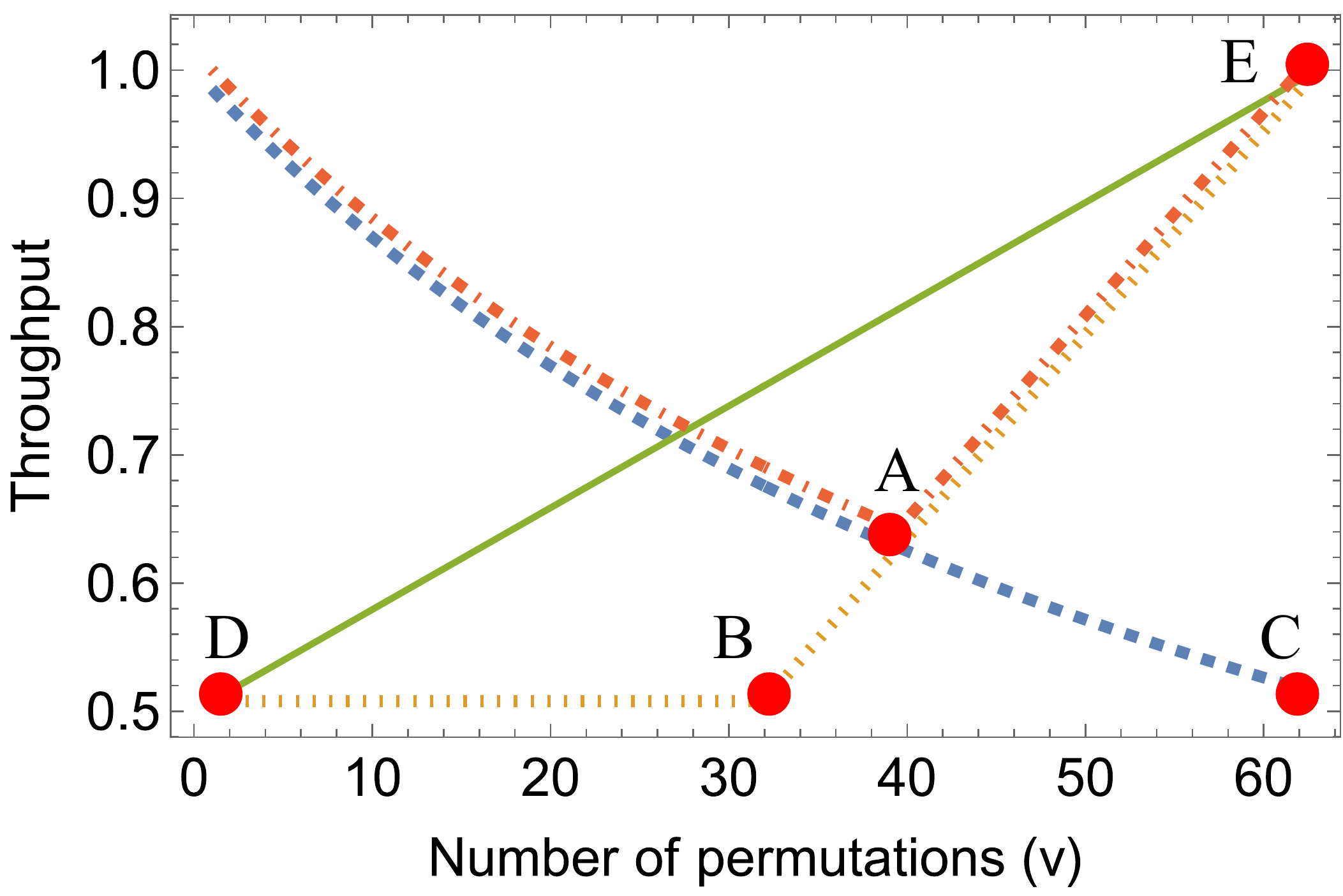}  \\
  \small{(a) The DCT for M(v)} & \small{The throughput for M(v)}
  \end{tabular}
   \caption[The demand completion time and throughput for the family M(v)]{The demand completion time and throughput for the Family of $M(v)=M(v,0)$ matrices. }\label{fig:dctBoundMV}
\end{figure*}

\section{A Case Study: The system DCT of \msn for the M(v,u) Matrices Family}\label{sec:thecasestudy}
 
In \autoref{sec:NetwrokDesignsCSN}, and \autoref{sec:NetwrokDesignsRSN}, we discussed the bounds and worst-case DCT of the systems \rsn and \csn for a general double stochastic matrix. However, no such bounds are known for \msn for a double stochastic matrix. That is, it is unknown what is the matrix with the worst DCT, and what this DCT is.
This section will discuss the DCT bounds of \rsn and \csn compared to our \msn system for a wide family of double stochastic traffic matrices $M(v,u)$. We leave the general question open for future research.
Recall that we first defined the $M(v,u)$ matrix family in \autoref{subsec:MvuMatrixDef}.

Throughout this evaluation,  we will use several assumptions: $\eff=1$, which means that the reconfiguration time $\rrec=0$. 
The transmission rate is normalized, meaning that the transmission rate is $r=1$.
For the numerical examples and demonstrations, we set $\crec=\frac{1}{n}$.  As we will see later, using this value means that both \rsn and \csn have the same system DCT.  This ensures no system has an obvious advantage over the other on any matrix from the set $M(v, u)$. Furthermore, a similar assumption is made in Eclipse \cite{bojja2016costlyEclipse}.
Lastly, since in this section, we assume that $n=64$, we get $\crec=\frac{1}{64}\approx 15ms=0.015s$.
We note that this value for $\crec$ is also a good approximation of the actual value of dynamic switch reconfiguration times \cite{memesF20011296,hall2021survey}.

Let us now consider the main result of this section, the system DCT for each of our systems in the following theorem.
This theorem describes the system DCT of all of our three systems of interest on the matrix family $\mathcal{M}=M(v,u)$.

\begin{theorem}[Systems DCT for $M(v,u)$]\label{thm:mainSystemDCTreslut}
 Let $\mathcal{M}$ be the family of all $M(v,u)$ matrices. Then,
 \begin{itemize}
   \item  The system DCT of \csn is:
\begin{align}\label{eq:sysDCTBVN}
     \dct(\bvn,\mathcal{M}) 
    = 1+\crec (n-1) 
\end{align}
which is about {\bf 1.94} for $n=64$ and $\crec=\frac{1}{n}$.
The worst-case is, for example, the matrix is $M(n-1,0)$.
  \item  The system DCT of \rsn is:
\begin{align}\label{eq:sysDCTRSN}
    \dct(\rr,\mathcal{M})= 2-\frac{2}{n} 
\end{align}
which is about {\bf 1.96} for $n=64$.
The worst-case matrix is $M(1,0)$, a permutation matrix.
  \item  The system DCT of \msn for a \bvn reconfiguration time $\crec < \frac{(2n-4)}{n^2}$:
\begin{align}\label{eq:sysDCTMSN}
    &\dct(\msn,\mathcal{M}) \le 
    \frac{\sqrt{1+4\crec(n -1)}+1}{2} 
     \end{align}
    which is about {\bf 1.6} for $n=64$ and $\crec=\frac{1}{n}$.
    The worst-case matrix is $M(\Ddot{v},0)$,
    where $\Ddot{v}=\frac{\sqrt{1+4\crec(n -1)}-1}{2\crec}$.
    \end{itemize}
For all systems the worst-case matrices are achieved for $u=0$, so $M(v,0)=M(v)$
\end{theorem}
Before diving into this theorem, let us first make two important observations.  
Consider how will \algPivot\ work with the $M(v)$ family of matrices. 
When \algPivot\ algorithm is used with $M(v)$, it converges into a simpler form. Since, in this case, there is only one $\beta_i$ size (as seen in the pseudo-code in \autoref{alg:piovt}), we will always send the entire matrix using either \csn or \rsn.
The DCT of this system and algorithm is the following observation for the $M(v)$, family.
 \begin{observation}[DCT of \msn with \algPivot\ on $M(v)$]\label{obs:naiveMixDCT}  
\begin{align}
    & \dctmix(\stalgmsn,\spalgmsn,M(v)  )= \nonumber\\
    &=\min( \dctda(\bvn,\direct,M(v) ), \dctrot(\rr, \upper,M(v)))
\end{align}  
\end{observation}
 This follows from the linearity of the DCT in \csn, and the fact that all permutations are of the same size.  
 
The second observation is the main takeaway of this section.
The system DCT \msn is strictly lower then \csn and \rsn.

\begin{observation}
\label{obs:sysDCTMSNSup} 
 Let $\mathcal{M}$ be the family of all $M(v,u)$ matrices.
 When $0 < \crec < \frac{(2n-4)}{n^2}$, then the system DCT of \msn is strictly lower than both \csn and \rsn.
\begin{align}
    & \dct(\msn,\mathcal{M}) < \min(\dct(\bvn,\mathcal{M}),\dct(\rr,\mathcal{M}))
\end{align}  
\end{observation}

Part of \autoref{thm:mainSystemDCTreslut} claims the worst case matrix for every system is found at $u=0$, to get an intuition of why this is the case for systems using \rsn (\rsn and \msn), we note that we can separate the $M(v,u)$ matrix into its two components, and schedule the uniform component using \direct\ separately and the non-uniform component using either \upper for \rsn or \algPivot\ for \msn. Since the uniform matrix $M( n-1)$ can be scheduled at an optimal rate using \direct. For \csn, we can still decompose $M(v,u)$ using $n-1$ matchings, thus, the system DCT is unaffected. We give a more formal proof of this point in the appendix in \autoref{app:worstCaseAnalysisMVU}.
 
\para{An Illustrative Example}
To get a glimpse of the results of \autoref{thm:mainSystemDCTreslut}, we present them for the case of $M(v,0)=M(v)$ in \autoref{fig:dctBoundMV}. In this figure, we present the general curves of the DCT of our three systems
to visualize the results better.
We present the behavior of the matrix DCT for the entire range of $M(v)$, i.e., $v=1, \dots, 63$.  
In \autoref{fig:dctBoundMV} (a), we plot the DCT for \csn (using \bvn and \direct), \rsn (using \rr\ and \upper), and \msn (using \algPivot). We additionally plot the lower bound for \rsn from \autoref{thm:lowerBoundRSN}, since the optimal DCT for each case is not known (See Appendix \autoref{sec:loweroundOfRSN} for short discussion). 
In \autoref{fig:dctBoundMV} (b), we see the results of the throughput of each system.
We add five significant points, $[A,B,C,D,E]$, to the figures and we discuss them next. Note that the points are homologous to both figures.
Looking at the figures, recall that \msn is using the \algPivot\ as in \autoref{obs:naiveMixDCT}, thus, the DCT of \msn can be described as the line representing the minimum between the DCT of the \upper algorithm in \rsn and the DCT of \csn using \bvn, where the two lines meet at point $(A)$. This point is the worst-case demand for \msn as seen in \autoref{eq:sysDCTMSN}.  
For \rsn the worst-case matrix is at point $(D)$, which is the permutation matrix as proven in \autoref{thm:LowerBoundPerm} and \autoref{lemma:feasibleApp} and stated in \autoref{thm:mainSystemDCTreslut} and \autoref{eq:sysDCTRSN}.
We note that for matrices between points $(D)$ and $(B)$, \upper has the same bound as in $(D)$.
Finally, for \csn the worst case point is at $(C)$, which, as we noted in \autoref{eq:sysDCTBVN} is for the uniform matrix where $v=n-1$. Note that the numerical results stated in \autoref{thm:mainSystemDCTreslut} fit the location of $(A)$, $(B)$, and $(C)$ in this figure. 
Looking at the \rsn lower bound, we see that it merges with the \rsn upper bound using $\upper$ in both of the extreme points, including the worst-case point at $M(1)$, which is marked by $(D)$, and best case matrix which is $M(n-1)$ at point $(E)$.
This proves that $(D)$ is the worst matrix for \rsn.
To conclude, the throughput (i.e., $\frac{1}{\dct}$) of both \csn and \rsn is about $0.5$, the throughput of \msn is $0.625$ about $25\%$ more, as seen in \autoref{fig:dctBoundMV} (b).


Recall that, with regard to \rsn and \csn, we have already proven the system DCT for general double stochastic matrix, in \autoref{thm:main:RRsysDCT} for \rsn, and in \autoref{obs:main:DAsysDCTBound}. Since the worst-case matrix for each system is part of the $M(v,u)$ matrix family, we focus on developing the results of theorem \autoref{thm:mainSystemDCTreslut} for \msn. 
For clarity, we consider the case $u=0$.
 \subsection{DCT of \msn on M(v)}\label{sec:DCTofMSNonMv}
 
In this section, we discuss the DCT of \msn the $M(v)$ type matrices. Since our definition for \msn includes the use of a two-component system, \bvn and \rsn using \upper algorithm, we first give exact expressions for the DCT of those systems on $M(v)$ and conclude with the DCT of \msn. 
This result, in turn, will help to compute the system DCT of \msn.

\para{The DCT of \csn on $M(v)$} For \csn, the result is very straightforward and based directly on previous results.
We can use the DCT from \autoref{eq:damatDCT}. We decompose the matrix $M(v)$ into $v$ permutation matrices. And since $\frac{ \Wcard{M(v)}}{r n}=1$ we can arrive at the result.  $\dctda(\bvn,\direct,M(v))=1+\crec v$.

\para{The DCT of \rsn for $M(v)$}
We state the DCT on $M(v)$ with \rsn in the following lemma.
\begin{lemma}[The DCT of \rsn on $M(v)$]\label{lemma:DCTRROnMvUpper}
The DCT of \rsn on the $M(v)$ matrix type is 
    \begin{gather} 
   \dctrot(\rr,\upper,M(v)) = 
\begin{cases}
   (2-\frac{2}{n}) & \text{for }1\le v<\frac{2}{n}\\  
  \frac{(n-1)}{ v }    & \text{otherwise}
\end{cases}
\end{gather}
\end{lemma}
\begin{proof}
Recall the result for the DCT of \rsn using algorithm $\upper$ in \autoref{def:rsnUpperBound}. This is the minimal value between the bound found in \autoref{thm:upperRSNBVN} and the bound found in \autoref{thm:defRotDCTGenApp}.
Setting the correct values for our setup in these two theorems, i.e., $\eff=1$, $r=1$, $\frac{\Wcard{M(v)}}{n}=1$ and $\max(M(v)_{i,j})=\frac{1}{v}$. We get two monotonous functions. This means we must find the (single) value of $v$ where they intersect. This is
$\frac{(n-1)}{ v }=(2-\frac{2}{n})$. Solving for $v$ we get $v\geq \frac{n}{2}$. The Lemma follows.
\end{proof}
 
We can now continue to the main
result of this section, the analysis of the system DCT of \msn.

\para{The DCT of \msn for $M(v)$}
First, we note that the general expression for the DCT of \msn is a two-part expression depending on the \bvn reconfiguration time $\crec$. However, in this section, we only present the part relevant to our assumptions, that is, for  $\crec < \frac{(2n-4)}{n^2}$ and we note that $\crec=\frac{1}{n}=0.015$ is in this range when $n=64$.
Furthermore, while we have set $\rrec=0$, for \csn, if $\crec=0$, the DCT will be $1$ for all matrices, as this is the only 'tax' or cost of \csn, making this system optimal for any demand.
Here we present the case $\crec < \frac{(2n-4)}{n^2}$, the  fully general case, including $\crec \geq \frac{(2n-4)}{n^2}$ is in the appendix in \autoref{app:thm:DCTMixMvuFUll}.
\begin{lemma}[DCT of \msn on $M(v)$]\label{lemma:msnUpperDCT}
For $u=0$ and $\crec < \frac{(2n-4)}{n^2}$
the DCT of \msn for $M(v)$ is the following expression.
\begin{gather}
   \dctmix(\stalgmsn,\spalgmsn,M(v))  = \\=
\begin{cases}
  1+\crec v & 1\le v<\frac{\sqrt{1+4\crec(n -1)}-1}{2\crec}\\
  \frac{(n-1)}{v}   & \text{otherwise}
\end{cases}
\end{gather}
\end{lemma}

\para{Proof sketch}
For $u=0$ and $\crec < \frac{(2n-4)}{n^2}$ the DCT of \csn is lower than the DCT of \rsn and \upper for $v<\frac{2}{n}$.
Following \autoref{obs:naiveMixDCT}, 
to prove the lemma, we only need to find the intersection point between the DCT of \csn, $1+\crec v$ and the DCT of \rsn where the \direct\ schedule archives the minimum, i.e., the DCT is  $\frac{n-1}{v}$. Solving for $v$ gives the result and ends the proof sketch.
To find the system DCTs as stated in  \autoref{thm:mainSystemDCTreslut}, we note that the worst-case matrix is when $v=\frac{\sqrt{1+4\crec(n -1)}-1}{2\crec}$ and the DCT is then $\dct(\msn,\mathcal{M})=1+\frac{\sqrt{1+4\crec(n -1)}-1}{2\crec}\crec=\frac{\sqrt{1+4\crec(n -1)}+1}{2}$.

To conclude, we see that our result and proof of \autoref{thm:mainSystemDCTreslut} show that \rsn and \csn have a system DCT about $22\%$ higher than \msn for our chosen parameters and the family $M(v,u)$.
However, this numerical result is only for our ad-hoc set of parameters. We may want to explore what happens to this ratio in other cases.
First, let us denote the ratio between system DCT  as $\psi$. This is the minimum between the ratio of the system DCT of \rsn to \msn and \csn to \msn, that is
\begin{align}\label{eq:sysDCTRatio}
   \psi= \min(\frac{ \dct(\rr,\mathcal{M})}{\dct(\msn,\mathcal{M})},\frac{\dct(\bvn,\mathcal{M})}{\dct(\msn,\mathcal{M})})-1
\end{align}
Where $1$ was reduced such that $0\leq \psi \leq 1$.

We make the following observation regarding our parameter set and the ratio of systems DCT.

\begin{observation}[Limit of $\psi$]\label{obs:optimalParmsMSN}
    Let $\mathcal{M}$ be the family of all $M(v,u)$ matrices.
   Assuming that the \bvn reconfiguration time is bounded by $\crec < \frac{(2n-4)}{n^2}$.
$\psi$  is maximized when $\crec=\frac{1}{n}$.
For this value of $\crec$ when the number of nodes $n$ increases 
    \begin{align}
        \underset{ n\rightarrow \infty}{ \psi} =\sqrt{5}-2\approx0.236 
    \end{align}
\end{observation}
 This can be deduced by finding the maximal value for $\psi$ from \autoref{eq:sysDCTRatio}, which achieved when $\dct(\rr,\mathcal{M})=\dct(\bvn,\mathcal{M})$, and as a result the \bvn reconfiguration rate must be set to $\crec=\frac{1}{n}$.
This observation, combined with our main result in \autoref{thm:mainSystemDCTreslut}, shows that $\msn$ maintains its advantage over the other systems for any system size.


In the following section, we empirically explore the performance of our three systems on a general double stochastic matrix. 

\section{Empirical Evaluation}\label{sec:Empirical}
In this section, we evaluate the performance of \msn and use \csn and \rsn as benchmarks. For \msn we consider the novel \algPivot, which is already described in \autoref{sec:msn}.
We begin by describing the traffic generation model and then present the results of our evaluation.  

\subsection{Traffic Generation}
Our traffic generation model is based on the model used in Eclipse \cite{bojja2016costlyEclipse} and in \cite{valls2021birkhoff} and denoted as $\TM(\largeRatio,\numberFlows,\largeLoad,n)$
A demand matrix is generated using the following four parameters. $\numberFlows$: the total number of flows per node, 
$\largeRatio$: the fraction of large flows per node, where $0 \le \largeRatio \le 1$, and 
$\largeLoad$: the load of large flow traffic, e.g., where if $\largeLoad=0.7$, then the large flows have $70\%$ of all traffic in terms of bits. 
Lastly, $n$ is the size of the square demand matrix. 
We define the number of large flows as $\numberFlowsLarge=\ceil{\largeRatio \numberFlows}$, and the number of small flows as $\numberFlowsSmall= \numberFlows-\ceil{\largeRatio \numberFlows}$.
The matrix is generated by adding $\numberFlows$ uniformly random derangement permutation matrices $P_i$. Each permutation is weighted by its type, whether small or large, plus additive white Gaussian noise to each permutation. 
The weight of a large flow permutation is, therefore, the following $\frac{\largeRatio}{\numberFlowsLarge}+\mathcal{N}(0,\sigma_l^2)$, where we set $\sigma_l= \lambda \frac{\largeRatio}{\numberFlowsLarge}$. We use a similar weight for the small flows where $\sigma_s= \lambda \frac{1-\largeRatio}{\numberFlowsSmall}$. Where $\lambda$ is the deviation from the mean weight. We use $\lambda=0.01$, meaning the standard deviation is about $1\%$ of the average.

This might generate a slightly smaller or larger matrix than a doubly stochastic matrix; therefore, we normalize its weight to one.

\subsection{Evaluation}
In this section, we discuss our empirical results. We will first see results on a chosen set of parameters and then show a sensitivity analysis. 
In the evaluation, we use the "birkhoff+" algorithm from the Birkhoff decomposition using the Julia package from  \cite{valls2021birkhoff}. Their \bvn algorithm creates what we described as an $\epsilon$ schedule, i.e., a schedule that leaves a small amount of leftover traffic that is not served. We define this formally in \autoref{def:epsilonSchedule}. 
Since we have used their algorithm for the evaluation, we also use the same $\epsilon$ value used in the evaluation of the paper, which is $\epsilon=10^{-4}$.
To keep results consistent, the matrix we used to evaluate \csn and \msn is a reconstructed matrix from the Birkhoff decomposition of  "birkhoff+".

\begin{figure*}[t]
  \begin{centering}
  \begin{tabular}{ccc}
  \includegraphics[width=.3\linewidth]{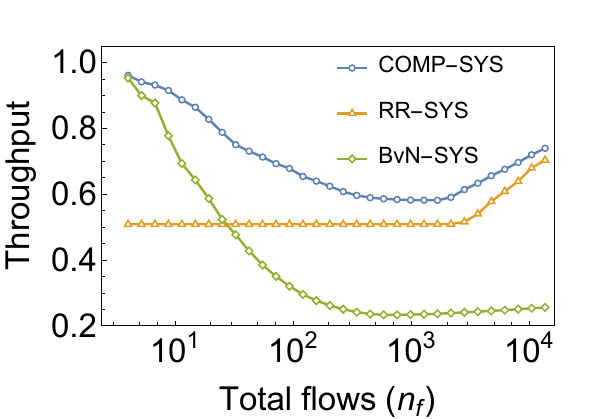} &
 \includegraphics[width=.3\linewidth]{ 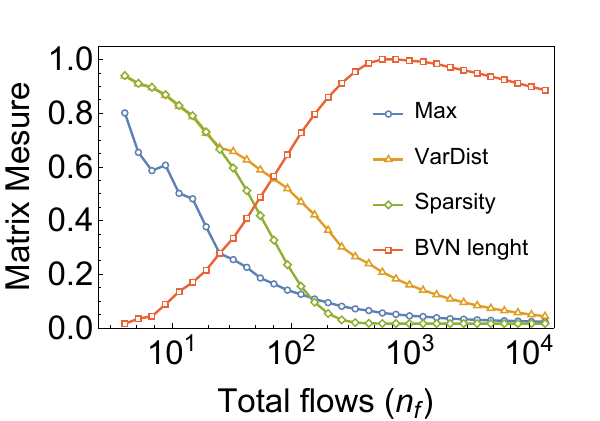}
 &
 \includegraphics[width=.3\linewidth]{ 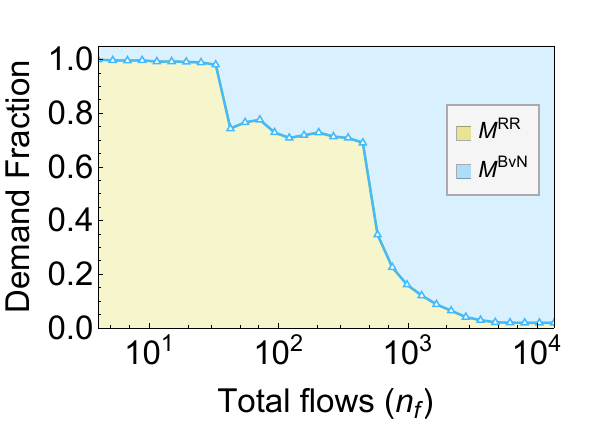}\\
  \small{(a) Three different systems } & \small{(b) Different matrix measures}  &  \small{(c) \algPivot division of the demand}
  \end{tabular}
   \caption[The throughput of the three systems on our traffic model and four different measures of the demand matrices used in the evaluation]{(a) The throughput of the three systems on matrices generated by the traffic model. Where \msn was tested with \algPivot. Note that \msn:\algPivot\ is slightly moved up. (b) Four different measures of the demand matrices from (a). (c) \algPivot\ division of the demand to $M^{\bvn}$ and $M^{\rr}$.
   }
    \label{fig:DCTsystems}
  \end{centering}
  \vspace{-.2cm}
\end{figure*}

Our main set of parameters is similar to previous work, such as in \cite{valls2021birkhoff,bojja2016costlyEclipse}. 
Unless otherwise stated, We set the fraction of large flows $\largeRatio$ to $20\%$, and the load of large flows $\largeLoad$ to $70\%$.

In \autoref{fig:DCTsystems} (a), we see the results of the three systems on traffic generated from our model. 
The $y$ axis represents the throughput, and the $x$ axis represents the number of flows generated per host.
In this simulation, we ran up to $4n^2$ 
flows per host and repeat the simulation $30$ times. 
We use $n=64$, and for \csn we use a reconfiguration time of $\crec=0.01s$.\footnote{We note that for the $\crec$ we use a similar value as the related work \cite{bojja2016costlyEclipse,valls2021birkhoff}, and that using this value in the previous analysis for $M(v,u)$ would have only improved the result further for \msn.}.
For \rsn we set $\eff=1$, that is $\rrec=0$.
In short, the traffic is generated from the model with parameters $\mathcal{T} = \TM(\largeRatio=0.2,\numberFlows=x,\largeLoad=0.7,n=64)$, where $x$ is our changing parameter and $4\le x \le 4n^2$.
First and foremost, we see the \msn consistently achieve higher throughput than \rsn and \csn. That is, the ability to decompose the matrix and serve it using two different types of schedules improves the throughput across the board.

We now use \autoref{fig:DCTsystems} (b) to help us explain the results. This figure shows the same $x$ axis, and on the $y$ axis, we see several normalized measures of the demand matrix. 
The figure presents i) Sparsity: the fraction of empty (zero) cells in the matrix. ii) Max: the value of the maximal element in the matrix. iii) BVN length:
the number of elements in the \bvn decomposition divided by the length of the longest decomposition among all those tested. iv) VarDist: the Variation distance of the matrix from the uniform matrix.
We observe using these measures that as we increase the number of flows, the generated matrices become more dense and uniform, and \bvn\ needs a larger number of permutations.
Looking at the \autoref{fig:DCTsystems} (a), on the curves of both \csn and \rsn, we can observe behavior that is similar in essence to the expected behavior from \autoref{fig:dctBoundMV} (a).  However, the crossing point between \rsn and \csn is very different here. 
The different traffic models can explain this difference in the results. In our current model, $\TM$, traffic flows are generated uniformly at random. As a result, some flows may overlap and increase the irregularity, thus decreasing the rate at which the maximal value is reduced.  However, we see that as more and more flows are introduced to the system, the more similar it becomes to $M(n-1)$ from the previous section.

Recall that \algPivot\ divides the original demand matrix $M$ into two sub-matrices $M^{\bvn}$ and $M^{\rr}$ each handles a different share of the traffic. 
In \autoref{fig:DCTsystems} (c), we see how this division is made on the demand matrices we generate. The $y$ axis shows the relative fraction of traffic scheduled on \csn or to \rsn by \algPivot, while the $x$ axis is the number of flows per host, as before.
Naturally, as less load is scheduled to \csn, more is scheduled to \rsn. 
It is possible to divide the figure into three regions: a region where traffic is almost exclusively scheduled with \csn (until about $50$ flows), a region where both systems are used significantly (until about $200$ flows), and a region where traffic is mostly scheduled with \rsn. 
We can see that the transition from one region to another is relatively abrupt. 
This can be explained in the following manner: in the first region, there are few \emph{small} flows; these flows are still relatively large, and thus, they are scheduled along with the large flows to \csn. In the second region, the result is slightly more predictable. We see that the load scheduled to \csn roughly fits with $\largeLoad$, the load of the large flows ($\largeLoad=0.7$) generated in the model. This implies that, in this region, the \emph{small} flows cross a threshold where they are small enough to be scheduled on \rsn more efficiently than with \csn. In the third region, more and more traffic is sent to \rsn, this implies that the matrix becomes more uniform as more flows are represented in the demand matrix, as we also see in \autoref{fig:DCTsystems} (b).
Comparing \autoref{fig:DCTsystems} (a) to (c), we might expect the DCT of \msn with \algPivot to be the same as \csn in the first region and the same with \rsn in the third region.
However, in the first region, there is a fraction of traffic, small in volume (related to the noise that we used in the traffic generation model), which \msn can schedule with \rsn. In contrast, \csn must schedule any traffic with the cost of a switch reconfiguration, which incurs a delay of $\crec$ per configuration regardless of the traffic volume sent. In the third region, the DCT of \rsn is becoming similar to that of \msn, however, if the \bvn decomposition can capture even a few large flows, scheduling them with \csn is still advantageous enough so that \msn has a lower DCT than either \csn or \rsn, on the entire tested range.

To conclude, we can see that the worst case throughput for each system is $\frac{1}{\dct(\msn,\mathcal{T})}\approx 0.58$ for \msn, $\frac{1}{\dct(\rsn,\mathcal{T})}\approx 0.507$ for \rsn, and for \csn  $\frac{1}{\dct(\csn,\mathcal{T})}\approx 0.232$.
The results for \csn are significantly different, which is expected as the total number of permutations rises, as we can also see in \autoref{fig:DCTsystems} (b)\footnote{We take note of the interesting phenomenon in this context: when the $\numberFlows$ increases, at some point, the length of the schedule starts to decrease. As more and more permutations are added to generate the demand, the closer the matrix is to the uniform matrix $M(n-1)$ (as is also supported by the variation distance), which can be decomposed to only $n-1$ elements, less than the worst case.}.

\begin{figure}[t]
  \begin{centering}
  \begin{tabular}{cc}
  \includegraphics[width=.47\columnwidth]{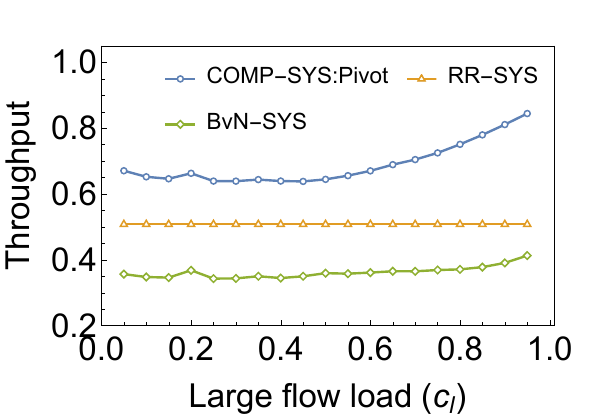} &
  \includegraphics[width=.47\columnwidth]{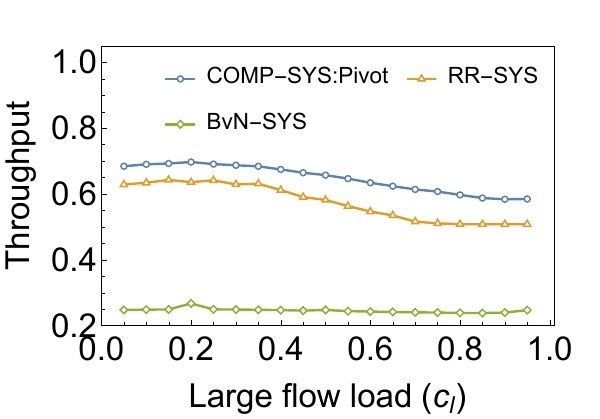} \\
  \small{(a) sparse workload} & \small{(b) dense workload}
  \end{tabular}
   \caption[The throughput on a dense and a sparse workload, with a change in the large flow load]{The throughput on a dense and a sparse workload, with a change in the large flow load $\largeLoad$}
    \label{fig:largeFlowLoad}
  \end{centering}
  \vspace{-.2cm}
\end{figure}

In the following part of this section, we study two other parameters of the model: the ratio (number) of large flows, $\largeRatio$, and the load of the large flows $\largeLoad$.
We consider two important workload types. \emph{sparse} workloads, where $\numberFlows=64$ and \emph{dense} workloads where $\numberFlows=3000$.
 In \autoref{fig:largeFlowLoad} (a) we test the following traffic model $\TM(\largeRatio=0.2,\numberFlows=64,\largeLoad=x,n=64)$, that is, a relatively sparse workloads and in \autoref{fig:largeFlowLoad} (b) and we test $\TM(\largeRatio=0.2,\numberFlows=3000,\largeLoad=x,n=64)$, a dense workloads.

In \autoref{fig:largeFlowLoad} (a) \& (b), we see that \msn achieves better throughput than either of the other systems regardless of the value of $\largeLoad$. In 
\autoref{fig:largeFlowLoad} (a) the throughput of \rsn is constant at $0.507$. This is because the matrices are sparse. In these kinds of matrices, there is a large gap between the maximal cell and average cell size, as seen in \autoref{fig:DCTsystems} (b).  We also see that both \msn and \csn improve with the larger $\largeLoad$. For \csn, since in this setting the number of permutations generated by the model is constant, we would expect no change in throughput; however, since the smaller permutations become so small when large flow becomes larger, they are discarded in the \bvn decomposition since we perform an $\epsilon$ schedule. 
Any discarded permutation will improve the completion time, but as we can observe in the figure, this effect is limited. For \msn we see a larger improvement than for \csn. Since the number of large permutations is static, the larger the load handled on these few permutations, the more effective \msn becomes, as this the system more effectively amortizes the reconfiguration time for \csn. In contrast, unlike \csn, it can handle the fraction of load represented in the small flows more effectively via \rsn.

For the dense workloads setting in \autoref{fig:largeFlowLoad} (b), we see that \csn is static, as the change in the load $\largeLoad$ doesn't affect the number of permutations. When the load $\largeLoad$ is small, these few matchings don't affect the shape of the matrix. Since the small flows are a majority in terms of numbers and load share, they form a more uniform matrix, which both \rsn and \msn can take advantage of in the same manner, hence their similar curves.

\begin{figure}[t]
  \begin{centering}
  \begin{tabular}{cc}
  \includegraphics[width=.47\columnwidth]{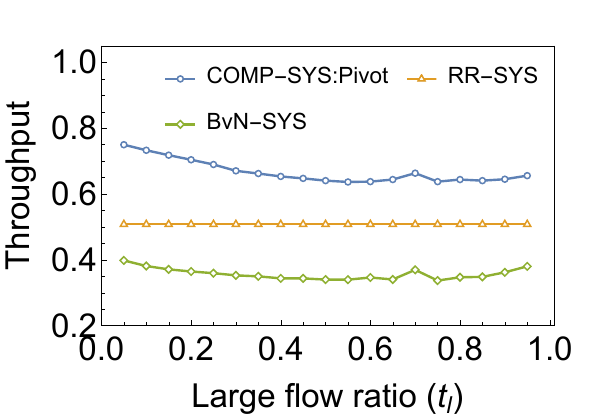} &
  \includegraphics[width=.47\columnwidth]{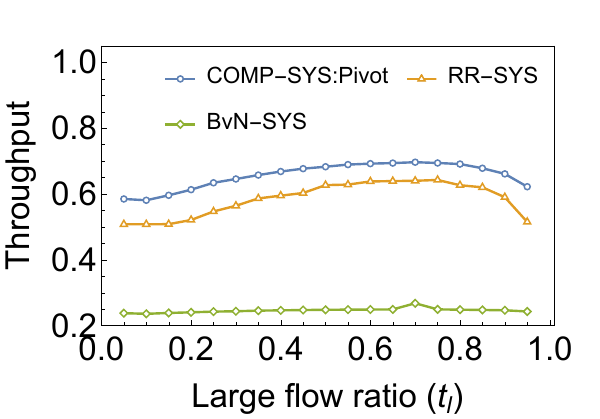} \\
  \small{(a) sparse workload} & \small{(b) dense workload}
  \end{tabular}
   \caption[The throughput on a dense and a sparse workload, with a change in the large flow ratio]{The throughput on a dense and a sparse workload, with a change in the large flow ratio $\largeRatio$}
    \label{fig:largeFlowRatio}
  \end{centering}
\end{figure}

In \autoref{fig:largeFlowRatio} (a) \& (b), we test what happens when we change the ratio of the number of large flows. In \autoref{fig:largeFlowRatio} (a)  we test the following traffic model $\TM(\largeRatio=x,\numberFlows=64,\largeLoad=0.7,n=64)$, sparse workloads, and in \autoref{fig:largeFlowRatio} (b) we test the following traffic model  $\TM(\largeRatio=x,\numberFlows=3000,\largeLoad=0.7,n=64)$, dense workloads.
For \autoref{fig:largeFlowRatio} (a), we see, as before in \autoref{fig:largeFlowLoad} (a), that for the sparse workload, the \rsn result is static.
But \msn and \rsn show a slight change, which seems very similar. This change results from the fact that at both extremes of this traffic model, where the ratio of large flows tends to one or zero, the matrix becomes more like a matrix with a single type of traffic and, therefore, generates fewer permutations. In \autoref{fig:largeFlowRatio} (b), we see the opposite effect. In the extremes of the traffic model, there are few of either large or small flows, which means that the matrix is less uniform and, therefore, the more uniform matrix benefits the throughput of \rsn. 

We can conclude that \msn is superior to either \rsn or \csn while changing the load and ratio parameter of the large flows. We see that for different workloads, different aspects of the \msn system become dominant, contributing to the change in throughput. Both \rsn and \csn showcase distinct strengths and weaknesses across various workloads. Leveraging their respective advantages, \msn can use them synergistically, resulting in better outcomes.

 \section{Related work}\label{sec:relwork}
Traditional datacenter network topologies are static and predominately based on different flavors of Clos or fat tree topologies \cite{closAl2008scalable,liu2013f10,jupiterSingh2015}, with some more recent proposals using expander graphs \cite{singla2012jellyfish,xpanderKassing2017beyond,zhang2019understandingFatClique}.  
More recently, Dynamic OCSs allowed the advent of reconfigurable datacenter networks 
 (RDCNs) \cite{wang2009your,hall2021survey}.

The two basic approaches to topology engineering in RDCNs are \emph{demand-aware} and \emph{demand-oblivious} topology designs~\cite{avin2019toward}. 
In a nutshell, demand-aware topologies are dynamic networks that adjust the topology to the actual demand, optimizing the topology to particular traffic patterns by rewiring connections and changing the topology.  
On the contrary, demand-oblivious topologies~\cite{hall2021survey} 
present the same topology to any traffic patterns, oblivious to the actual demand.  Note, however, that demand-oblivious systems can be dynamic (e.g., \rotornet \cite{mellette2017rotornet}) or static (e.g., Clos \cite{closAl2008scalable}).

 \para{Demand-aware networks and BvN type Systems}
Typical examples of demand-aware topologies are Eclipse \cite{bojja2016costlyEclipse}, projecTor \cite{ghobadi2016projector}, Morida  \cite{moridafarrington2013multiport}, and others \cite{zhou2012mirror,solsticeliu2015scheduling,hamedazimi2014firefly}. 
The topology in these systems can change dynamically. This topology is a union of matching generated from a schedule of matchings where all or most traffic flows are served in a direct connection, i.e., single hop. 
To construct the topology schedule, some works treat this problem as a max weighted matching problem on bipartite graphs, such as in Helios  \cite{farrington2010helios}, or an online stable matching problem, such as in projecTor. In Eclipse   \cite{bojja2016costlyEclipse} and FireFly  \cite{hamedazimi2014firefly}, the authors use Birkhoff von Neumann (BvN) decomposition to create a schedule of matchings, which we denoted in this paper as \csn type systems. 

While BvN saw a wide array of applications in several fields in the context of dynamic networks, it was often used to generate a schedule of matchings tailored to the demand of the network \cite{mckeown1999islip,chang2002load,hamedazimi2014firefly,bojja2016costlyEclipse,solsticeliu2015scheduling,livshits2018lumos}. In the context of dynamic networks, to be efficient, a schedule needs not only to cover all the traffic but also to use a minimal number of switch reconfigurations, which is a known NP-hard problem \cite{BVNNPdufosse2016notes, BVNNPkulkarni2017minimum}.  In a recent paper, \cite{valls2021birkhoff}, Valls et al. attempt to optimize \bvn and construct a schedule in a less naive way than previous works, allowing for fewer configurations (fewer matchings) and faster run time. Indeed, the run time of the generation process for a novel schedule itself could have a non-negligible effect on the reconfiguration time, which was also improved in \cite{valls2021birkhoff}, over simpler methods used in Eclipse \cite{bojja2016costlyEclipse} and Solstice \cite{solsticeliu2015scheduling}.
 
\lp{In Solstice \cite{solsticeliu2015scheduling}, the authors define the problem of scheduling for a hybrid network of electrical packet switches and optical circuit switches. They use an algorithm based on \bvn to generate matching for heavy flows and send smaller "flows" to the electrical network. They assume that the matrix is sparse and skewed, and this assumption helps to bound the number of matchings found by \bvn. In some sense, each matching here has a duration equal to the maximal entry since, in the first phase, they start to "stuff" the matrix until it is bistochastic. They greedily choose the matchings and stop finding new matches when the remaining traffic is small enough to be sent to the packets switch. 
Eclipse \cite{bojja2016costlyEclipse} presents a similar approach to Solstice. The authors also employ Multi-Hop Routing between each matching on the graph formed from the matchings. 
In Lumos \cite{livshits2018lumos} the authors study the trade-off between keeping the current configuration and the reconfiguration time.
Similarly to Solstice, they find the maximum cardinality matching, that is the matching with the largest minimum element. This should achieve better circuit utilization than Eclipse, using maximum weighted matching. }

The reconfiguration time is an important factor limiting demand-aware reconfigurable networks, which is limited by both hardware and software. The hardware (i.e., the OCS) must deal with moving physical mirrors, vibrations causing interference, etc. \cite{ghobadi2016projector}. The 3D MEMES switch ~\cite{memesF20011296}) allows, in principle, arbitrary reconfigurations of the topology in periods on the scale of \emph{micro-seconds}.
From a software perspective, algorithms need to generate a topology schedule, collect demand \cite{wang2009your}, etc., computations that might be significant compared to the physical reconfiguration time, as we have mentioned regarding BvN.  
 
\para{Demand Oblivious networks and round-robin type Systems}
Other systems, which we termed in this paper as \rsn (e.g., \rotornet), can achieve  \emph{nono-seconds} scale reconfiguration time. They achieve this mainly by using a predefined topology schedule of matchings and going through it periodically in a round-robin manner. In this paper, we assumed that the union of the matchings in one period of the topology schedule is a complete graph; however, this is not always the case, such as in Mars \cite{addanki2023mars}, where the graph could be of a lower degree.
\lp{In early precursors for round-robing type systems such as \rotornet and similar works,
\cite{chang2002load}, the authors show what they call a "Birkhoff-von Neumann" (\bvn) switch. This switch abstraction uses two phases. In the first, a load-balancing algorithm will attempt to generate uniform demand for the second phase. In the second phase, they use \bvn to generate $n$ (Or $O(n)$) matchings and then rotate between them. The matchings could be, for example, a set of circular permutation matrices. But they choose the duration of each matching according to \bvn.}
In Sirius \cite{ballani2020sirius}, \rotornet \cite{mellette2017rotornet} and Opera \cite{mellette2020expandingOpera}, the schedule generates full mesh connectivity over time, therefore, packets may be sent directly to any node in one hop. However, this may not be very efficient in every case, for example, when the matrix is skewed. In a skewed matrix, load-balancing methods are employed. 
For example, in Sirius, all packets are sent at two hops, while for \rotornet and Opera, only packets that cannot be sent at a single hop are sent at two hops. Furthermore, \rotornet uses RotorLB \cite{mellette2017rotornet}.  A decentralized protocol that allows nodes to negotiate routing decisions. Briefly, in RotorLB, nodes only help other nodes if they have available capacity, thus spreading the load more evenly.
In \cite{wilson2023extending}, the authors discuss a more generalized version of round-robin networks. In this work, rather than each node participating in a round-robin schedule between all nodes, each node participates in a partial schedule, where only a part of the nodes are connected. This allows the system to achieve better latency at a cost of lower throughput.

\para{Composite systems}
This paper considered the composite system a more flexible reconfigurable network design, which we denoted as  \msn. These systems combine the topology engineering approaches of both round-robin-like systems and BvN demand-aware-like systems.
A close example of a \msn type system is Cerberus~\cite{griner2021cerberus}. In Cerberus, the authors combine several network designs, rotor, demand-aware (not based on BvN), and, to a lesser extent, a static expander into a single unified system. The paper shows that throughput can be increased when traffic is routed towards the most appropriate sub-network. However, Cerberus did not offer an accurate, in-depth analysis of the bounds on the performance of its composite system in terms of demand completion time. 
\lp{In the paper CacheNet~\cite{griner2022cachenet}, the authors also combine \rotornet and a demand-aware system based on the operation of a computer cache. However, the topology in CacheNet is not a k-regular graph, and their measurements differ from those used in this work. }

\para{Performance metrics}
As we discussed in this section, there exists a large body of work on the design of communication networks, different metrics on which to compare the performance of different network designs \cite{yuan2014lfti,jyothi2016measuring,namyar2021throughput}, e.g., bisection bandwidth and other cut-based metrics. However, these give little information regarding the worst-case throughput performance \cite{jyothi2016measuring}.
In their work, Jyothi et-al. \cite{jyothi2016measuring}, defines throughput, in short, as the maximal scale-down factor for a demand matrix (traffic matrix), for which the matrix the feasible, a definition also used in some more recent works \cite{griner2021cerberus, addanki2023mars,namyar2021throughput}.  The authors note that finding the worst-case matrix for a given network is difficult. Therefore, they present ways to find a near-worst-case traffic matrix for several known systems. 
In the more recent paper by Namyar et al. \cite{namyar2021throughput} further explores the relation between bisection bandwidth and throughput. The authors also note that the throughput of reconfigurable demand-aware systems is an open problem. We note that our paper is supposed to be a step in this direction.

\section{Discussion}\label{sec:discution}
This paper proposes a simple model to compare and study the tradeoff between demand-aware and demand-oblivious topologies and traffic schedules.
Our model can be seen as an extension of the well-known \bvn-decomposition problem and an attempt to integrate Topology and Traffic engineering in a simple model.

We propose \msn and show that it can achieve higher throughput for a wide range of model parameters. Nevertheless, many open questions remain. First and foremost, what is the worst-case (double stochastic) demand matrix for this new type of system? Other interesting directions are how to extend the results to multiple spine switches and/or smaller size switches (matching). Also, can the throughput improve with even more than two-hop routing? We leave these questions for future work.

\bibliographystyle{IEEEtran} 
\bibliography{bibi.bib}

\appendices

\section{Schedules, Appendix to Section \ref{sec:modelAndProbDefinition}}\label{app:ps}
In this section, we state the properties of a feasible traffic schedule and a few helpful definitions regarding the traffic schedule needed for our proofs. The first definition is a 
of matrices that present the total indirect and direct traffic that is sent to each link at each time slot.

\begin{definition}[Indirect and direct total matrices]\label{app:def:IndirectDemandMat}
For a given traffic schedule $\ps$, with a subset of permutation list $\{\configElem_1,...,\configElem_v\}$, we define the $i$th indirect local demand matrix  $\mohSum_i$ as an $n \times n$ matrix with the following property.

 \begin{gather*}
\mohSum_i[k,j]:= 
\begin{cases}
\sum_{l=1}^n \moh_i[k,l] & \text{if }\configElem_i[k,j]=1\\  
  0    & \text{otherwise}
\end{cases}
\end{gather*}
Similarly, for a direct non-local demand matrix  $\mthSum_i$.
 \begin{gather*}
\mthSum_i[k,j]:= 
\begin{cases}
\sum_{l=1}^n \mth_i[l,j] & \text{if }\configElem_i[k,j]=1\\  
  0    & \text{otherwise}
\end{cases}
\end{gather*}
The direct local demand matrix is the same as the equivalent direct local traffic schedule matrix, 
\begin{align*}
    \mdlSum_i:=\mdl_i
\end{align*}
\end{definition}
 We can redefine the Total traffic matrix of \autoref{def:Totalmatrix} using the above matrices. 

\begin{definition}[Total traffic matrix]\label{app:def:Totalmatrix}
Given a demand matrix $M$, and \ps schedule, we define the total traffic matrix $\totalDemandMat$ as the total amount of bits that were sent on each link,
\begin{align*}
\totalDemandMat  = \sum_{i=1}^v\mdlSum_i+\sum_{i=1}^v\mohSum_i+\sum_{i=1}^v \mthSum_i
\end{align*}
\end{definition}
That is, it's the sum of all induced demand from direct and indirect traffic.

\subsection{Feasible Traffic Schedule Properties}\label{app:subsec:feasibltySchedule}

A \emph{feasible} traffic schedule needs to hold several properties. Recall that a detailed traffic schedule is a sequence of communication instructions: $\sentry_i=\{t_i,w_i,s_i,d_i,x_i,y_i\}$.
 
\para{property \#1 - admissibility}
For each $\sentry_i$, let $P_{t_i}$ be the active switch configuration at time $t_i$ then $P_{t_i}[x_i,y_i]=1$.

\para{property \#2 - time sufficiency}
Transmission time is sufficient on each configuration.
$\forall 1 \leq k,l \leq n$: 
\begin{align*}
    \alpha_{i} \ge \frac{\hat{m}_i[k,l]}{r}
\end{align*}

\para{property \#3 - causality}
For every traffic that is sent over 2-hops, every first hop must be followed by a second hop at a later time and every second hop must have a first hop at an earlier time. 

It follows from property \#3 that $\mth_1=0_n$ and $\moh_v=0_n.$

Another definition used in previous works is the $\epsilon$ schedule, a feasible but non ``complete'' schedule.

\begin{definition}[$\epsilon$ schedules]\label{def:epsilonSchedule}

Given a demand matrix $M$, we denote a pair of \ts and \ps 
as $\epsilon$ schedules if the schedules are feasible, but send all traffic but an $\epsilon \ge 0$ Frobenius fraction of it.
Formally,
\begin{align*}
\Forbcard{M - \left( \sum_{i=1}^v\mdl_i + \sum_{i=1}^v\moh_i \right ) }  \le  \epsilon 
\end{align*}
Where $\Forbcard{\cdot}$ is the Frobenius norm. 
\end{definition}

\section{Proofs for Sections \ref{sec:NetwrokDesignsRSN}}\label{app:sec:mainProofSectionForRSNDCT}
In this section, we prove the bounds we discussed in \autoref{sec:NetwrokDesignsRSN}.
We first prove the bound for the direct traffic scheduler for \rsn.

\subsection{Proofs for \autoref{thm:LowerBoundPerm}} \label{sec:LowerBoundPerm}



We begin by defining a useful concept, the \emph{traffic skewness} of $\ps$, denoted as  $\skew(\ps)$ (or $\phi$ for short), which we note was already defined similarly in Cerberus \cite{griner2021cerberus}. 
However, here, we will see how it can be naturally and formally defined using our traffic schedule framework.  
In short, $\skew(\ps)$ denotes the fraction of bits in \rsn, which are sent through a single hop, and $1-\skew(\ps)$ is the fraction that is sent via two hops. 
\begin{definition}[The skewness of a traffic schedule]\label{def:skew}
Given a \ps the skewness of the schedule $\skew(\ps)$ is the ratio between directly sent traffic to all sent traffic, formally
    \begin{align}
     \skew(\ps) &= \frac{\Wcard{\Mdl}}{\Wcard{\Mdl}+\Wcard{\Moh}+\Wcard{\Mth}}
    \end{align}
\end{definition}

When the traffic schedule is \emph{complete}, we can claim the following.

\begin{lemma}[The skewness of a complete traffic schedule]\label{lem:completeSkew}
Let \ps be a complete traffic schedule (using Definition \ref{def:Totalmatrix}), then,
    \begin{align}
        \skew(\ps) &=2-\frac{\Wcard{\totalDemandMat}}{\Wcard{M} }
    \end{align}
Where $M$ and $\totalDemandMat$ are the demand matrix and the total traffic matrix, respectively.
\end{lemma}

\begin{proof}
Since $\skew(\ps)$ also can also be used to denote the average (per bit) path length for $\ps$. As $\skew(\ps)$ is the fraction of bits sent at one hop and $1- \skew(\ps)$ is the fraction sent at two, the average path length will be \begin{align*}
    1\cdot \skew(\ps) + 2\cdot(1- \skew(\ps)) = 2- \skew(\ps)
\end{align*}
In turn, since $\ps$ is complete, every bit in $M$ is sent via one or two hops, so the size of the total traffic matrix $\Wcard{\totalDemandMat}$, is a function of the average path length that is
\begin{align*}
    \Wcard{\totalDemandMat}=(2- \skew(\ps)) \Wcard{M}
\end{align*}
Moving terms gives us the lemma.
\end{proof}

Note that, by definition  
$0\leq \skew(\ps) \leq 1$. Since we use either one or two-hop traffic. 
In the rest of the paper, for simplicity, we will assume that the \ps is implicit and denote $\skew(\ps)$ as simply  $\skew$.

Another definition we state here is the \emph{inactive cell}.
Not every algorithm we consider will use all cells in the matrix to transmit data.  We define an inactive cell in the following manner.

\begin{definition}[Inactive cell]\label{def:inactiveCell}
Given a traffic schedule \ps, and the total traffic matrix formed from it, a cell of $\totalDemandMat$ with the coordinates $[k,l]$ is inactive if it's never used to transmit traffic or  formally (using Definition \ref{app:def:IndirectDemandMat}),  
\begin{align}
   \forall 1 \leq i \leq v,  \mthSum_i[k,l]+ \mohSum_i[k,l]+  \mdlSum_i[k,l]=0
\end{align}
That is, this cell is zero in all the total traffic schedule matrices. 
\end{definition}

One example of such cells is the cells of the diagonal, which in our model are always inactive. 

We denote the set of all inactive cells in row $i$ as $\inActCell_i$, and the total number of these in row $i$ as $\card{\inActCell_i}$. Often, we will assume that the number of inactive cells in all rows is the same. We will denote this number as $\inActCellMax$. That is $\forall i \in [1...n], ~~\card{\inActCell_i} = \inActCellMax$.

We note that since the total traffic matrix results from a traffic schedule, we can use it to determine a lower bound on the DCT. This is given by looking at the maximal element in $\totalDemandMat$.

\begin{claim}[Total traffic matrix DCT]\label{app:clm:TotalmatrixDCT}
Consider any demand matrix $M$ and any algorithm \spalg.
Let $\totalDemandMat$ be the total traffic matrix of the traffic schedule. 
The DCT of \rsn, is lower bounded by the maximal entry in $\totalDemandMat$ 
 \begin{align}
    \dctrot(\rr ,\spalg,M ) \geq (n-1)\frac{\max(\totalDemandMat_{i,j})}{\eff r}
\end{align}  
\end{claim}
\begin{proof}
Since traffic in the total demand matrix is sent only directly (one hop)
the DCT could be determined by using the \direct\ traffic scheduler on the total demand matrix $\totalDemandMat$, similarly as in \autoref{thm:defRotDCTGenApp}. 
\end{proof}
Note that this is a lower bound, as, at best, the maximal element of the total demand matrix $\max(\totalDemandMat_{i,j})$ is transmitted continuously, without pause, but this may not always be the case.
When $\spalg$ generates a schedule that transmits continuously on all connections, this is an exact result, that is, the DCT for such a schedule is the transmission time of the maximal entry.


To find the lower bound on the DCT of \rsn, we prove the following theorem.  
\begin{restatable}{theorem}{lowerBoundRSN}[Lower Bound of DCT in \rsn]\label{thm:lowerBoundRSN}
 Let $M$ be a demand matrix and $\spalg$ a traffic scheduling algorithm for \rsn that generates a complete \ps with $\inActCellMax$ inactive cells in each row. 
 Then the \rsn demand completion time $\dctrot$ is lower bounded as
\begin{align}\label{eq:rotDCTGenInactive}
    \dctrot(\rr,\spalg ,M ) \geq  (2-\skew)\frac{1}{ \eff r} \frac{ \Wcard{M}}{ n}   \frac{ n-1}{n-\inActCellMax} 
\end{align}
 
\end{restatable}

  \begin{proof}
  We reach the results using the following derivation, which we explain below,
\begin{align}
    \dctrot(\rr,\spalg,M) &\ge    \frac{n-1}{\eff r} \max(\totalDemandMat_{i,j}) \label{eq:max} \\
    &\ge \frac{n-1}{\eff r} \frac{\Wcard{\totalDemandMat_}}{n \cdot (n-\inActCellMax)} \label{eq:average}\\
    &= (2-\skew)\frac{1}{\eff r} \frac{ \Wcard{M}}{n}\frac{n-1}{n-\inActCellMax} \label{eq:phi}
\end{align}
Eq. \eqref{eq:max} follows from \autoref{app:clm:TotalmatrixDCT} that a  lower bound on the DCT is the size of the maximal cell in the total traffic matrix $\totalDemandMat$. 
Eq. \eqref{eq:average} is true since the largest cell
is lower bounded by 
the average \emph{active} cell size (active cells are all cells that are not inactive).  
That is, this is $\max(\totalDemandMat_{i,j})\ge\frac{ \Wcard{\totalDemandMat}}{(n-\inActCellMax)n}$. Since \ps is complete, we can use our result for $\skew$ from \autoref{lem:completeSkew} in Eq. \eqref{eq:phi}. That is, we use the following result, $\Wcard{\totalDemandMat}=(2- \skew) \Wcard{M}$, which gives us the last step of the theorem.
\end{proof}

Note that the result from \autoref{thm:lowerBoundRSN} gives no bounds on the value of $\skew$. 
Next, we prove an upper bound on $\skew$ for the $M(v)$ matrix family (Section \ref{sec:Preliminaries}), which we will use to prove the DCT.

\begin{theorem}[Upper bound of $\skew$ on a M(v)]\label{thm:lowRSNApp}
Consider a demand matrix $M(v)$. Let $\ps^*$ be any complete traffic scheduling for \rsn with $\inActCellMax$ inactive cells in each row, which minimizes the DCT for $M(v)$. Than $\skew^*=\skew(\ps^*)$ is upper bounded by,
\begin{align}
\skew^* \leq \skew_{\min}=\frac{2v}{(n-\inActCellMax+v)}
\end{align}
 \end{theorem}

\begin{proof}
Consider a complete schedule and its skewness $\skew$.
From \autoref{lem:completeSkew} we have $\Wcard{\totalDemandMat}=(2- \skew) \Wcard{M}$.
From \autoref{app:clm:TotalmatrixDCT} it follows that 
the lower bound for the DCT is minimized when the maximum cell in $\totalDemandMat$ is minimized. 
Next, we claim that since the schedule is optimal, each cell in $M(v)$ that has an entry $\frac{1}{v}$ will only send \emph{direct} traffic. This follows from the observation that if it was needed to send an \emph{indirect} traffic that travels two hops, it was better to send the same amount of traffic directly, reducing the overall traffic sent. Contradicting the optimality of the schedule.

Therefore, there are two types of cells in the total demand matrix $\totalDemandMat$ : $nv$ cells that send direct traffic and at most $n(n-w-v)$ cells that send \emph{indirect} traffic. The total amount of direct traffic 
is $\skew \Wcard{M} = \skew n$. The maximal cell that sends direct traffic is at least the average cell $\frac{n\skew}{nv}$.
Similarly, the total amount of indirect traffic 
is $2(1-\skew)\Wcard{M} =2(1-\skew)n$.
The maximal cell that sends indirect traffic is at least the average cell $\frac{2(1-\skew) n}{n(n-\inActCellMax-v)}$.
Therefore, the minimum of the maximum cell in $\totalDemandMat$ is achieved for $\skew_{\min}$,
\begin{align*}
    \skew_{\min} = \argmin_{\skew} \max(\frac{\skew}{v},\frac{2(1-\skew)}{n-\inActCellMax-v}).
\end{align*}
The minimum $\skew$ is then when 
\begin{align*}
    \frac{\skew_{\min}}{v} = \frac{2(1-\skew_{\min})}{n-\inActCellMax-v}
\end{align*}
Which result in 
\begin{align*}
    \skew_{\min} = \frac{2v}{n-\inActCellMax+v}
\end{align*}
\end{proof}

We can now prove \autoref{thm:LowerBoundPerm}.

\LowerBoundPerm*
\begin{proof}
   To prove the lower bound, we first set the number of inactive cells to $\inActCellMax=1$, that is, only cells on the diagonal are inactive. 
   now, we can state the lower bound for \rsn as a function $\skew$ as
   \begin{align*}
       &\dctrot(\rr,\spalg ,M ) \geq (2-\skew)\frac{1}{ \eff r} \frac{ \Wcard{M}}{ n}
   \end{align*}
Clearly since $\frac{1}{ \eff r} \frac{ \Wcard{M}}{ n}>0$ we need to find an upper bound on $\skew$.
This bound is also true for matrices from the double stochastic matrix family $M(v)$.

Since, as we recall, the only inactive cells are on the diagonal, $\inActCellMax=1$, we can state the upper bound for $\skew$ on the $M(v)$ matrix family as
\begin{align*}
    \skew  \leq \frac{2v}{(n-1+v)}
\end{align*}
Setting $v=1$ in the $M(v)$ matrix family results in the permutation matrix family. That $P_\pi=M(1)$, and thus, for any $P\in P_\pi$ we get
\begin{align*}
    &\dctrot(\rr,\spalg ,P ) \geq (2-\frac{2}{n})\frac{1}{ \eff r} \frac{ \Wcard{M}}{ n}
\end{align*}
This concludes the proof for the theorem. 
\end{proof}

\subsection{ The \algOnePerm Scheduler}\label{app:sec:AlgOnePerm}

\begin{figure*}[t]
  \begin{centering}
  \begin{tabular}{c}
 \includegraphics[width=1.0\linewidth]{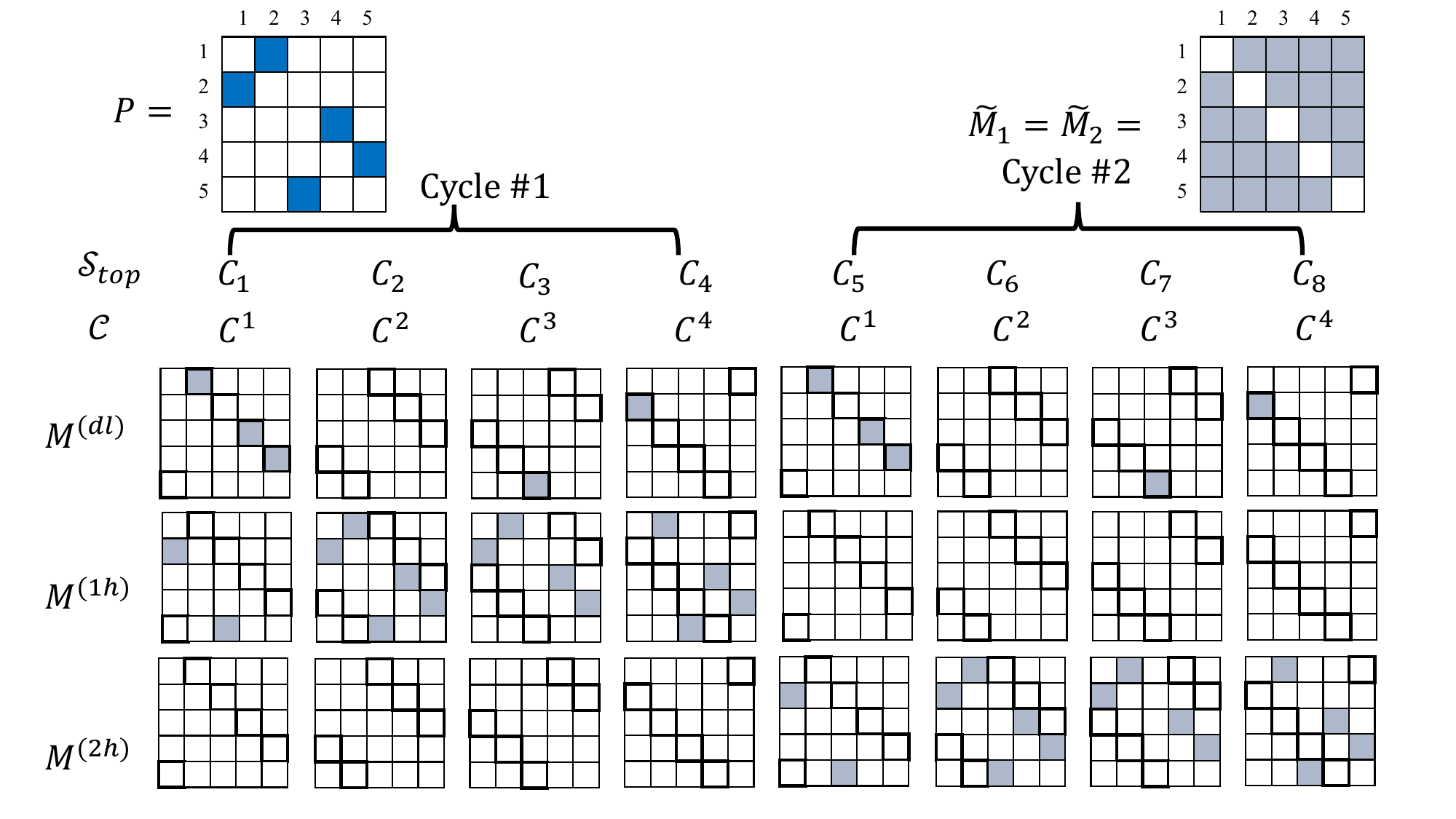}  
  \end{tabular}
   \caption[An example of the traffic schedule generated from \algOnePerm for a small matrix]{
  An example for a $5\times 5$ demand permutation matrix $P$, and the traffic schedule generated for it, \ts. The traffic schedule $\ts=\{\configElem_1,\configElem_2,\configElem_3,\configElem_4\}$ is the same topology schedule as used in \autoref{fig:example}.  Each of the two cycles used in the schedule is marked. They generate the same total traffic matrix for each cycle $\totalDemandMat_i$. Where each element marked in grey is equal to $\frac{1}{n}=\frac{1}{5}$, and the elements of the permutation matrix are marked by blue, are equal, and represent a traffic of $1$  }
    \label{fig:EDPSschdual}
  \end{centering}
\end{figure*}

In this section, we will discuss the Permeation Scheduler or, as we denote it, $\algOnePerm$. This algorithm generates a detailed traffic schedule for the permutation matrix, which, as we will later see, is optimal.  In \autoref{alg:onePremSched}, we see the pseudo-code for the scheduler. 
This scheduler receives a single permutation demand matrix $P$ and a 
cycle configuration $\cycleSched$ for \rsn
and uses them to generate a detailed traffic schedule for the demand.
We note that a detailed traffic schedule could be summarized as in \autoref{subsec:ScheduleDefinition}.
This topology schedule is composed of two cycles, that is $v=2n-2$. The first $n-1$ elements in $\ts$ belong to the first cycle, and the last $n-1$ elements belong to the second. Lastly, recall that the union of permeation matrices of the cycle configuration $\cycleSchedElem^i \in \cycleSched$ is a complete graph.

Before we begin describing the algorithm, we define a tool that will be useful in the pseudo-code of the algorithms.
For a permutation matrix $P$
we denote with the vector $\pi^{(P)}=\{\pi_1,\pi_2,...,\pi_n\}$ the indices of each non zero element in the matrix $P$, where for each row $k$ in $P$, it holds that $P[k,\pi_k]=1$. 
 
The traffic schedule generated from \algOnePerm is composed of two equal ``phases'' each corresponding to a single cycle, in turn, corresponding to the two nested \emph{For} loops in the pseudo-code which start at line \ref{line:alg:perm:FirstC} and line \ref{line:alg:perm:SecC} of the pseudo-code.  
Since $P$ is a permutation matrix, for each row $k$ in $P$ we have one cell with demand (whose indexes are denoted by  $\pi_k$) and $n-2$ empty cells, which could be used for indirect routing (other than those on the diagonal). 
Inside the two nested \emph{For} loops, we see how the algorithm redistributes $\frac{(n-2)}{n}$ bits of traffic from the demand cell in each row towards the empty cells. 
For example, in the first nested \emph{For} loop in line  \ref{line:alg:perm:directFirstC}, we see that the algorithm generates an entry for direct traffic, this will happen exactly $n$ times per run of the nested \emph{For} loop, once per each element in $P$. In line \ref{line:alg:perm:indirectFirstC}, we see an indirect local traffic entry for the empty cells of which there will $n(n-2)$ instances. This is similar to the second cycle. However, here instead of generating $n(n-2)$  indirect local entries, we generate $n(n-2)$  direct non-local traffic entries. 
We can see that for all entries in the schedule, their weight ($w$) is set to $\frac{1}{n}$. Since there are $n-2$ empty cells and one demand cell on each row, each cell in the first and last cycle will send $\frac{1}{n}$ bits of data together each cell will send $\frac{2}{n}$ bits during the course of the entire schedule.  

To help us understand the result of this algorithm, we can first summarize the detailed traffic schedule as in \autoref{subsec:ScheduleDefinition}, and then look at it from a cycle perspective. To generate the first total demand \emph{cycle} matrix, which we denote as $\Tilde{M}_1$, we take the first $n-1$ indirect and direct demand matrices and summarize them similarly to \autoref{app:def:Totalmatrix}.
\begin{align*}
    \Tilde{M}_1 =\sum_{i=1}^{n-1}\mdlSum_i+\sum_{i=1}^{n-1}\mohSum_i+\sum_{i=1}^{n-1} \mthSum_i
\end{align*}
We construct  $\Tilde{M}_2$ similarly, using the following set of $n-1$ matrices in the schedule.
Since the algorithm will send  $\frac{1}{n}$ bits on each connection once per cycle, we know that each element in each matrix is equal to $\frac{1}{n}$. Therefore, each of the total demand cycle matrices is of the form  $\Tilde{M}_1=\Tilde{M}_2=\frac{1}{n}(J_n-I_n)$. Where $J_n$ is an all ones $n\times n$ matrix and $I_n$ is the $n\times n$ identity matrix. 
In \autoref{fig:EDPSschdual}, we see an example of a schedule generated from algorithm \algOnePerm on a $5\times 5$ permutation matrix $P$. The traffic schedule is composed of a set of four permutation matrices, $\{\cycleSchedElem_1,\cycleSchedElem_2,\cycleSchedElem_3,\cycleSchedElem_4\}$,
the same we use in \autoref{fig:example}, they are denoted by the thicker border inside the schedule matrices.Looking at the two equal total traffic matrices for the two cycles $\Tilde{M}_1$, $\Tilde{M}_2$, we see that eventually, \algOnePerm distributes the permutation evenly among all connections.
The schedule itself will use the first cycle to send all indirect local traffic, $\Moh$, and keep $\mdl_i=0_n$, when $i\leq4$, and the second to send direct nonlocal traffic, $\Mth$, and keep $\moh_i=0_n$, when $i\geq4$. Both cycles are used to send direct local traffic $\Mdl$. We further note that the colored squares of each matrix in the schedule  $\Mdl$, $\Moh$, $\Mth$, and in $\Tilde{M}_1$, $\Tilde{M}_2$ represent a (normalized) demand of $\frac{1}{5}$ bits. 
After describing the algorithm, we prove the following result for \algOnePerm.

 \feasibleApp*
\begin{proof}
To prove this bound, we first prove the schedule generated is complete, and then we prove it is feasible, and lastly, we calculate the DCT of this schedule. 
For the purpose of brevity and simplicity, we will assume that $P$ is saturated and double stochastic and that the cycle hold time is $\eff=1$, and that the transmission rate is normalized $r=1$, that is, we set $\frac{\Wcard{P}}{\eff n r}=1$. 

First, recall that there are two cycles in the $\ps$ generated by the algorithm, which means we have $2n-2$ elements in $\ps$ in total.  We show that $\ps$ is complete since it satisfies \autoref{def:completeSchedule} of a complete schedule.
Summarizing the total direct and indirect local traffic matrices as in \autoref{subsec:ScheduleDefinition}  we get
\begin{align*}
&\sum_{k=1}^{2n-2}\mdl_k+ \sum_{k=1}^{2n-2}\moh_k=P
\end{align*}
\khen{see}
To see why this is the case let us look at the first nested \emph{For} loop in the pseudo-code of \autoref{alg:onePremSched} in line \ref{line:alg:perm:FirstC}. We see that the loop runs exactly once for each different connection of the first cycle, as this loop is used once for each $\cycleSchedElem^i$ and once per connection in each $\cycleSchedElem^i$. In total, for each node $k$,  we will have $(n-1)$ entries ($\sentry$) generated. $1$ entry for the direct traffic and $(n-2)$ entries for indirect local traffic. 
In the second \emph{For} loop in the pseudo-code of \autoref{alg:onePremSched} in line \autoref{line:alg:perm:SecC} we similarly have $1$ entry per node of direct traffic and zero indirect local traffic in this cycle.  Since the weight of each entry is $\frac{1}{n}$ we have a total weight of $\frac{1}{n}+\frac{1}{n}+\frac{n-2}{n}=1$ per node, as each line of $P$. Since every entry is generated with the following source and destination fields: $s=k$ and $d=\pi_k(P)$, which are the same as $P$ for every node $k$, that is, for every line in $P$, we have our proof for completeness.


We must look at the properties described in \autoref{app:subsec:feasibltySchedule} to show feasibility.
To prove the first property, property \#1, we need to prove that each entry $\sentry$ is admissible.
Since we built the schedule using only cells where $\cycleSchedElem^i[k,j]=1$, for all $\cycleSchedElem^i\in \cycleSched$, that is, only for connected matching, as we also set $t=i$ for every traffic entry $\sentry$,  as we see in lines~\ref{line:alg:perm:directFirstC},~\ref{line:alg:perm:indirectFirstC},~\ref{line:alg:perm:indirectSecC} ~and~\ref{line:alg:perm:indirectSecC},
  of the pseudo-code of \autoref{alg:onePremSched}, this propriety thus holds trivially.  

For property \#2,
In lines~\ref{line:alg:perm:directFirstC},~\ref{line:alg:perm:indirectFirstC},~\ref{line:alg:perm:indirectSecC} ~and~\ref{line:alg:perm:indirectSecC}, of \autoref{alg:onePremSched} we see that each entry has exactly $\frac{1}{n}$ bits. Since the size of the transmission time for each switch configuration is also $\frac{1}{n}$ seconds (recall that $r=1$) property \#2 will hold with one assumption: That no two entries are sent on the same connection at same switch configuration $\cycleSchedElem^i$.
To see if this is indeed the case, let us look at the first nested \emph{For} loop in the pseudo-code of \autoref{alg:onePremSched} in line \ref{line:alg:perm:FirstC}. We see that the loop runs exactly once for each connection of the first cycle, as this loop is used once for each $\cycleSchedElem^i$ and once per connection in each $\cycleSchedElem^i$. In total, we will have  $n(n-1)$ entries $\sentry$ generated.  So, the number of entries is correct, but is each entry sent on a different connection during the cycle? Again, looking at the first nested \emph{For} loop, each entry has a different source and destination (as $k$ and $j$ are always different) we see this is also correct. 
We can use the same arguments for the second nested \emph{For} loop of the second cycle in line \ref{line:alg:perm:SecC}. 
Thus we can conclude that each connection will be used exactly once for each switch configuration to transmit $\frac{1}{n}$ bits, and therefore, property \#2 holds.  

The rest of the properties can be proven naively by looking at the construction of the schedule from two cycles.
For propriety \#3 
this is true since, indeed, as we can see in the pseudo-code of \autoref{alg:onePremSched}, we see that all indirect local traffic is sent first. Therefore, this property holds, as there is exactly one packet sent in the first cycle, as in the second with the same weight.

Since the schedule is complete and feasible, we only need to calculate the duration of each cycle. Since, as we learn by looking at the total demand cycle matrices $\Tilde{M}_1$, and $\Tilde{M}_2$ each cell will send $\frac{1}{n}$ (normalized) bits in a cycle, we set $\delta$ to be $\delta=\frac{1}{n}$ seconds and each cycle has length of $(n-1)\frac{1}{n}$ for a total of 

\begin{align*}
    2(n-1)\frac{1}{n}=2-\frac{2}{n}
\end{align*} 

We can re-scale this expression by $\frac{\Wcard{P}}{\eff n r}$.
This means that  each cell will send $\frac{1}{n}\cdot \frac{\Wcard{P}}{n}$ bits, and we need to set $\delta$ to be $\delta=\frac{1}{n}\cdot \frac{1}{\eff r}$ seconds, to get our theorem.
\end{proof}

\begin{algorithm}[t]
  \caption{The \algOnePerm\ Algorithm}\label{alg:onePremSched}
  \begin{algorithmic} [1]
   \Require  $P \in {\mathbf{R}^+}^{n\times n}$, $\cycleSched$ \Comment{{\color{blue} saturated permutation matrix and a $\rr$ cycle configuration, $\cycleSched=\{\cycleSchedElem^1,\dots. \cycleSchedElem^{n-1}$\} }}
   \Ensure $\ditPS$\Comment{{\color{blue} A detailed traffic schedule (Section \ref{subsec:ScheduleDefinition})}}
   \State $\ditPS=\{\}$
  \ForAll {$\cycleSchedElem^i\in \cycleSched$}  \Comment{{\color{blue}  Generate the first cycle}} \label{line:alg:perm:FirstC}
   \ForAll {$k,j \in \cycleSchedElem^i[k,j]=1$} 
     \If{$j= \pi_k(P)$}   
      \State $\sentry=\{i,\frac{1}{n},k,\pi_k(P),k,\pi_k(P)\}$  \Comment{{\color{blue}  Direct traffic}} \label{line:alg:perm:directFirstC}
      \Else 
     \State $\sentry=\{i,\frac{1}{n},k,\pi_k(P),k,j\}$
      \Comment{{\color{blue}  Indirect traffic}} \label{line:alg:perm:indirectFirstC}
      \State  $\ditPS=\ditPS.\textsc{append}( \sentry)$ 
     
       \EndIf
    \EndFor
      \EndFor

   \ForAll {$\cycleSchedElem^i\in \cycleSched$}  \Comment{{\color{blue} Generate the second cycle}} \label{line:alg:perm:SecC}
   \ForAll {$k,j \in \cycleSchedElem^i[k,j]=1$} 
     \If{$j= \pi_k(P)$}
          \State $\sentry=\{i,\frac{1}{n},k,\pi_k(P),k,\pi_k(P)\}$   \Comment{{\color{blue}  Direct traffic}} \label{line:alg:perm:directSecC}
      \Else 
    \State  $\sentry=\{2i,\frac{1}{n},k,\pi_k(P),j,\pi_k(P)\} $ 
      \State  $\ditPS=\ditPS.\textsc{append}( \sentry)$   \Comment{{\color{blue}  Indirect traffic}} \label{line:alg:perm:indirectSecC}
       \EndIf
    \EndFor
      \EndFor
     \end{algorithmic}
\end{algorithm}
\label{alg:CycleMatchingNum}
 \subsection{ The \rrbvn Scheduler}\label{app:sec:rrbvnscheduler}
 
In this section, we describe the \rrbvn traffic scheduler algorithm.
This can schedule any doubly stochastic matrix $M$ and achieves a DCT, which we described in \autoref{thm:upperRSNBVN}. We give pseudo-code for this algorithm in \autoref{alg:rrbvn}.
The algorithm requires a demand matrix and a $\rr$ topology schedule ($n-1$ permutation matrices whose union creates a complete graph), The algorithm first creates a BvN decomposition from $M$, that is, it creates a permutation matrix list, $\{P_M=\{P_1,...,P_v\}$, and a coefficients list $\beta_M=\{\beta_1,...,\beta_v\}\}$. 
The algorithm then generates 'sub-schedules' $\ps^*$ for each permutation matrix using the \algOnePerm algorithm.  However, \algOnePerm requires a doubly stochastic matrix, which is why \rrbvn will normalize each matrix $P_i$ by $\beta_i$, that is, it will schedule $\frac{1}{\beta_i}P_i$. After receiving the schedule from \algOnePerm, the algorithm restores the demand to its original size. This can be done by, for example, multiplying each element in the traffic schedule by $\beta_i$, e.g. $\beta_i\cdot\moh_i$.
The algorithm then schedules each permutation sequentially, each time setting the slot hold time to $\delta_i=\frac{\beta_i}{r}$, where $\delta_i$ is the slot hold time for the $i$th sub-schedule.
Since the sum total of all coefficients is one, $\sum_v \beta_i=1$ the total DCT for any $M$ on \rsn is the same as the DCT for a permutation matrix (as we've proved in \autoref{thm:upperRSNBVN}.

\begin{algorithm}[t]
   \begin{algorithmic}[1]
     \caption{The \rrbvn\ Algorithm}\label{alg:rrbvn}
 
     \Require  {$M \in {\mathbf{R}^+}^{n\times n}$, ~$\cycleSched$ }\Comment{{\color{blue} a doubly stochastic matrix and a $\rr$ cycle configuration}} 
      \Ensure $\ps$ \Comment{{\color{blue} a traffic schedule}} 
        \State $\ps=\{\}$
     \State  $P_M=\{P_1,...,P_v\}$,$\beta_M=\{\beta_1,...,\beta_v\}$  \Comment{{\color{blue} Decompose $M$ to a set permutation matrices and coefficients using \bvn}} 
  
    \ForAll {$i \in \card{P_M}$}
      \State $\ps^* =\algOnePerm (\frac{1}{\beta_i}P_i) $
      \Comment{{\color{blue}  Create a schedule for $\frac{1}{\beta_i}P_i$ using \algOnePerm}}
      
      \State re-scale the schedule the schedule $\ps^*$  by $\beta_i$ (i.e., By multiplying each element in the traffic schedule by $\beta_i$)\;
      \State $\ps \cup \ps^*$
    \EndFor
    \end{algorithmic}
\end{algorithm}

\subsection{Proofs for \autoref{thm:defRotDCTGenApp}}\label{sec:directBoundProof}

\defRotDCTGenApp*

\begin{proof}[Proof of Theorem \ref{thm:defRotDCTGenApp}]
Recall that \rsn rotates between $n-1$ predefined matchings, which, when combined in a full cycle, form a complete graph. During a \emph{slot}, there is a single active matching on which packets are sent. 
Since each source-destination pair appears only once in the complete graph, it has to wait for all other $n-2$ matchings to rotate until it is available again to send packets in a single hop.  
We start from the most general expression for DCT as stated in \autoref{eq:scheduleGeneral}, where $\spalg=\rr$ and $\stalg=\direct$, we get.
\begin{align*} 
    \dctrot(\rr ,\direct,M )&= \sum_{i=1}^v (\alpha_i  +R_i) \nonumber
    \\ &=\sum_{i=1}^v (\delta  +\rrec) =v (\delta  +\rrec)
\end{align*}
We want to find a closed-form formula for this result without $v$.
Since, in this setting, we cannot send packets at more than one hop to transmit $M$, we need to wait for the network to send each entry on the demand matrix individually. Therefore, the DCT will be a function of the maximal entry of $M$, that is $\max(M_{i,j})$. 
Let us denote the number of complete cycles used to transmit $M$ as $x$.
In each cycle there are $n-1$ matchings use, and therefore $n-1$ periods of length $\delta$ used for transmission and $n-1$ reconfigurations of length $\rrec$ in each cycle, that is the DCT is,
$ \dctrot(\rr ,\direct,M )=x(n-1)(\delta  +\rrec)$.
Now, we need to find $x$. The number of cycles is the number of periods of length $\delta$ needed to transmit $\max(M_{i,j})$. This is simply the time to transmit the entire entry at full rate $\frac{\max(M_{i,j})}{r}$, to get $x$, we divide by $\delta$, that is,$x=\frac{\max(M_{i,j})}{\delta r} $. We get
\begin{align*}
    &\dctrot(\rr,\direct,M )=v(n-1)(\delta+\rrec)=\\
     &=\frac{(n-1)\max(M_{i,j})(\delta+\rrec)}{\delta r}
\end{align*}
And by setting $\eff=\frac{\delta}{\delta+\rrec}$ the proof is complete. 
\end{proof}

\section{The Running Time of \algPivot}\label{sec:sruuningtime}
The running-time complexity of \algPivot\ is foremost a function of the \bvn decomposition used, denoted as $O(\bvn)$ (and there exist many results in the literature as in \cite{valls2021birkhoff}).
Given a decomposition, we need to sort the $v$ elements of the decomposition, $O(v\log(v))$. To calculate the DCT of \csn, we need to look at the list of coefficients $\beta_i$ with $O(v)$.  To calculate the DCT of \rsn on its fraction of traffic, we need to find the maximal and mean element of each sub-matrix. We, therefore, need to construct sub-matrices, adding or removing one permutation at a time, which, since a permutation is sparse, is an $O(n)$ operation. This operation can be combined with simple bookkeeping to keep track of the maximal and mean elements. 
The run time of the for loop is $ vn+v$  in total. So the running-time complexity of the algorithm is $O(v(n +\log(v)))+O(\bvn)$.

 \section{Proofs and Analysis for Section \ref{sec:thecasestudy}}\label{app:sec:caseStudyTheromproof}

In \autoref{sec:thecasestudy}, we analyzed the system DCT for \msn on the $M(v)$ matrix. In this section, we expand this analysis to the more general case $M(v,u)$, and for all ranges of $\crec$ and $n$, we prove that the worst-case system DCT for all systems is at $u=0$ we also state the expressions for the other two systems on $M(v,u)$.


Before starting this section, we discuss simple extensions to \algPivot\ and the \upper traffic schedulers, denoted as \algPivotU and \upperU, which will allow both to work optimally with the $M(v,u)$ matrix family when $u>0$. 
For the systems \rsn and \msn, we will see that we can decompose the matrix  $M(v,u)$ into two parts $uM(n-1)$ and $(1-u)M(v)$, and schedule each part independently with no negative effects on DCT. 
We first define \algPivotU as a pair of topology and traffic schedulers $\{\stalgmsnU,\spalgmsnU\}$.
In this algorithm, we remove the uniform matrix component, $uM(n-1)$, from the demand matrix (if such exists) and then schedule it to be served using \rsn.
We then use \algPivot\ to schedule the remaining demand matrix, $(1-u)M(v)$. Since  $M(v,u) = (1-u)M(v) + uM(n-1)$ we serve all traffic.
We define the DCT of such a system and algorithm as follows
 \begin{definition}[DCT of \msn with \algPivotU]\label{def:naiveMixDCTUniform}  \begin{align}
    & \dctmix(\stalgmsnU, \spalgmsnU, M)=\nonumber\\
    & =\dctrot(\rr, \upper, uM(n-1))+ \nonumber\\&+\dctmix(\stalgmsn, \spalgmsn, (1-u)M(v))
\end{align}  
\end{definition}
Similarly, algorithm  \upperU, decomposes the matrix  $M(v,u)$ into two parts $uM(n-1)$ and $(1-u)M(v)$ and schedules the uniform part optimally on \rsn with the $\direct$ algorithm, while it schedules  $(1-u)M(v)$ with $\upper$ as the traffic schedule. 
We note that when $u=0$, both algorithms are the same as \algPivot\ and \upper.

We also note that with \csn it would not be helpful to decompose the matrix in such a way. Since \csn cannot use \rsn type schedules, this will increase the number of reconfigurations and the DCT.

\subsection{Worst Case Analysis for M(v,u)}\label{app:worstCaseAnalysisMVU}
In this section, we analyze the worst case for each system and show that it is at $u=0$ as stated in \autoref{thm:mainSystemDCTreslut}.

To prove the worst case for \csn, we use \bvn to decompose the matrix $M(v,u)$ into $n-1$ permutations.  $M(v, u)$ can be decomposed optimally into $v$ permutations, each with a total load per row of $\frac{1-u}{v}+\frac{u}{n-1}$, and another set of $n-1-v$ permutations each with a total load per row of  $\frac{u}{n-1}$ combining these, will give us.
\begin{align}
 &\dctda(\bvn,\direct,M(v,u))=\nonumber \\
 &=(\frac{1-u}{v}+\frac{u}{n-1})v+ \frac{(n-1-v)(u)}{n-1}+(n-1)\crec\nonumber \\ &=1+(n-1)\crec 
\end{align}
Note that here we assume a perfect \bvn decomposition algorithm with full knowledge of the original $v$ permutations. 

Since we have shown that there is no difference between the case of $M(v,u)$ for $u>0$ and $1\leq v\leq n-2$ and $M(n-1)$, this shows that for \csn, the worst case is found (also) at $u=0$.

For \rsn we prove the following lemma.
\begin{lemma}[worst case DCT for \rsn]\label{lem:RsnWorstCaseDCT}
    The worst case DCT for \rsn is for $u=0$. 
        \begin{align}
        \dctrot(\rr,\upperU, M(v,u))\leq  \dctrot(\rr,\upper ,M(v ))
    \end{align}
\end{lemma} 
\begin{proof}

With $\upperU$ for any $u>0$ we decomposes $M(v,u)$ to $M(v,u)=(1-u)M(v)+uM(n-1)$.
We then send each sub-matrix separately on \rsn. 
For any $v \ge 1$ and $0\leq u \leq 1$  is holds that  
 \begin{align*}
 &\dctrot(\rr,\upperU, M(v,u))=\\
  &= \dctrot(\rr,\upper,uM(n-1))+\dctrot(\rr,\upper,(1-u)M(v)) \\
  &\leq\dctrot(\rr,\upper,uM(v))+\dctrot(\rr,\upper,(1-u)M(v)) \\
  &= \dctrot(\rr,\upper,M(v))
 \end{align*}

\noindent  The third line holds since: 
\begin{align*}
\footnotesize
    &\dctrot(\rr,\upper,uM(n-1)) \leq  \dctrot(\rr,\upper,uM(v)),
\end{align*}
which it true since $M(n-1)$ is sent optimally in with a DCT of $1$ for any $u$ and $n$.

\noindent The fourth line holds since \upper will use the same algorithm (i.e., \direct or \rrbvn) in $M(u,v), (1-u)M(v)$ and $uM(v)$ since the ratio between the maximum and the average cells is the same in all three matrices. In turn, the DCT equality follows.
\end{proof}

Finally, for \msn we prove the following lemma.

\begin{lemma}[worst case DCT for \msn]\label{lem:mixNetWorstCaseDCT}
    The worst case DCT for \msn is for $u=0$.  
    \begin{align*}
        &\dctmix(\stalgmsnU,\spalgmsnU,M(v,u) \leq \dctmix(\stalgmsn,\spalgmsn,M(v))
    \end{align*}
\end{lemma}

\begin{proof}
With $\spalgmsnU$, for any $u>0$ we decomposes $M(v,u)$ to $M(v,u)=(1-u)M(v)+uM(n-1)$. 
We can send $uM(n-1)$ optimally on \rsn, incurring a DCT of $u$ seconds. We get,
 \begin{align*}
        &\dctmix(\stalgmsnU,\spalgmsnU,M(v,u) =  \\ 
        &=u+ \dctmix(\stalgmsn,\spalgmsn,(1-u)M(v))
    \end{align*}
For $(1-u)M(v)$, recall that we have two options in accordance with \autoref{obs:naiveMixDCT}: either we send $M(v)$ wholly with \rsn or with \csn.
When using \rsn, we have a similar case as with \autoref{lem:RsnWorstCaseDCT}, since the matrix is sent entirely on \rsn, the lemma follows for this case. 
When using \csn.
 the DCT for the submatrix $(1-u)M(v)$ is $1-u+\crec v$. 
 In total DCT for $M(v,u)$ on \msn in this case, therefore $u+1-u+\crec v=1+\crec v$, the same as the general case in \autoref{lem:RsnWorstCaseDCT}.
As a result, the DCT will remain the same. Therefore, the lemma follows.
\end{proof}
Following these lemmas, we note that the worst-case demand for all types of systems we study is $u=0$.

\subsection{System DCT of {\msn} on M(v,u)}
In this section, we finalize the proof for \autoref{thm:mainSystemDCTreslut}.

We begin by stating the generalized theorem for the DCT of \msn on $M(v,u)$ and the system DCT for the case of $\crec \ge \frac{(1-u)(2n-4)}{n^2}$. This theorem and its proof, therefore, will complete the result of \autoref{thm:mainSystemDCTreslut} for other values of $\crec$.
 
\begin{theorem}[System DCT of \msn on $M(v,u)$ on the entire range of $\crec$]\label{app:thm:DCTMixMvuFUll}

Given $M$ an $n\times n$ matrix $M\in M(v,u)$, the DCT of 
  \msn is the following two-part expression partitioned by the value of the reconfiguration time $\crec$.
  
\para{First case} $\crec < \frac{(1-u)(2n-4)}{n^2}$
\begin{gather}
   \dctmix(\stalgmsnU,\spalgmsnU,M) =\nonumber \\=
\begin{cases}
   1+\crec v & 1\le v<(1-u)\frac{\sqrt{1+4\crec(n -1)}-1}{2\crec}\\
  u+ (1-u)\frac{(n-1)}{v}   & \text{otherwise}
\end{cases}
\end{gather}

\para{Second case} $\crec \ge \frac{(1-u)(2n-4)}{n^2}$

\begin{gather}
   \dctmix(\stalgmsnU,\spalgmsnU,M)= \nonumber \\= 
\begin{cases}
  1+\crec v & 1\le v< \frac{(1-u)(n-2)}{n\crec}\\
 u+ (1-u) (2-\frac{2}{n})   &\frac{n-2}{n\crec} \le v <\frac{n}{2} \\
 u + (1-u)  \frac{(n-1)}{v}   & \text{otherwise}
\end{cases}
\end{gather}

Now, let $\mathcal{M}$ be the family of all $M\in M(v,u)$ matrices.
The system DCT of \msn is given when $u=0$, and the first case was already stated in \autoref{thm:mainSystemDCTreslut}, for completeness, we restate it here.

 \para{First case} $\crec \le \frac{(2n-4)}{n^2}$,
the worst case is given for $v=\Ddot{v}$ where $\Ddot{v}=\frac{\sqrt{1+4\crec(n -1)}-1}{2\crec}$.
\begin{align}
    &\dct(\msn,\mathcal{M}) \le \dctmix(\stalgmsnU,\spalgmsnU,M(\Ddot{v}))=\nonumber \\&=1+\frac{\sqrt{1+4\crec(n -1)}-1}{2\crec}\crec
\end{align}

\para{Second case} $\crec> \frac{(2n-4)}{n^2}$,
the worst case is given for $v=\Dot{v}$ where this value is $\Dot{v}=\frac{n-2}{n\crec}$.
For this value of $v$ the DCT is the same as the worst case DCT of \rsn from \autoref{thm:mainSystemDCTreslut}.
\begin{align}
    &\dct(\msn,\mathcal{M}) \le \dctmix(\stalgmsnU,\spalgmsnU,M(\Dot{v})) =\nonumber\\
    &=\dctrot(\rr,\upper,M(\Dot{v}))=2-\frac{2}{n} 
\end{align}
\end{theorem}
\begin{proof}

First, we show how to find the expressions for the DCT of \msn, we then prove what is the system DCT in each case.
In \msn we decompose into two matrices.
Using the \algPivotU, since \rsn can handle $M(n-1)$ optimally, we always send this sub-matrix using \rsn. 
For the other $M(v)$ matrix, use the \algPivot algorithm.

To define the DCT of \msn in a piecewise formula, we need to find the $v$ value where the different components of \msn intersect and then use the appropriate expression for their DCT.
The DCT for \msn on $M(v,u)$ is therefore
\begin{align}\label{eq:DCTofMSNGen}
       &\dctmix(\stalgmsnU,\spalgmsnU,M(v,u  ) ) = \\
    &= u+\dctmix(\stalgmsnU,\spalgmsnU,(1-u)M(v))
\end{align}
Note that since the total traffic load here is $1-u$ the DCT for each subsystem is slightly different.  
For systems with a different load size, $x$ (i.e., $x \neq 1$), 
the DCT of \csn on $xM$ where $M$ is a double stochastic demand matrix with a \bvn decomposition of length $v$ is, $\dctda(\bvn,\direct,xM)=x+v\crec$. While for \rsn we get $\dctrot(\rr,\upper,xM)=x\dctrot(\rr,\upper, M)$, we, therefore, place these expressions in a piecewise expression as a function of $v$ and $\crec$.
We can use \autoref{fig:dctBoundMV} as a visual aid in our explanation, as the curves we are interested in remain similar in principle.
To find the crossing point for \rsn with \csn, we have two cases: either the \csn crosses the \rsn, i.e., the $\upper$ curve at  at $u+(1-u)(\frac{n-1}{v})$  (that is, between point $(B)$ to $(E)$), or $u+(1-u)(2-\frac{2}{n})$ (that is, between point $(D)$ to $(B)$) . 
This crossing point area is a function of $\crec$, which changes for the \csn curve with $(B)$ being the point of discontinuity. At this point the \bvn reconfiguration time is $\crec = \frac{(1-u)(2n-4)}{n^2}$.

In the The first case for $\crec$ values smaller than $\crec\le \frac{(1-u)(2n-4)}{n^2}$. 
We denote the value of $v$ for the cross point where $1+\crec \Ddot{v}=(1-u)\frac{n-1}{\Ddot{v}}$ as $\Ddot{v}$. After solving for $\Ddot{v}$ we get $\Ddot{v}=(1-u)\frac{\sqrt{1+4\crec(n-1)}-1}{2\crec}$

The second case happens for $\crec > \frac{(1-u)(2n-4)}{n^2}$. The \csn DCT curve will cross the \rsn DCT curve at $u+(1-u)(2-\frac{2}{n})$.
We denote the value of $v$ for the second crossing point as $\Dot{v}$, here we get $1+\crec \Dot{v}=(1-u)(2-\frac{2}{n})$. After solving for $\Dot{v}$ we get  $\Dot{v}=(1-u)\frac{n-2}{n\crec}$.

Opening the expression in \autoref{eq:DCTofMSNGen} using our formulas for the DCT of \rsn and \csn in a piecewise manner and using our different cases for $\crec$ will give us the first part of the proof \autoref{app:thm:DCTMixMvuFUll}, namely, the expressions for the DCT. We now continue with proving the system DCT of \msn.

Since the worst case is at $u=0$ (as shown in \autoref{app:worstCaseAnalysisMVU}), for \msn the worst-case matrix is a function of $n$ and $\crec$.
We have two cases as a function of the \bvn reconfiguration time $\crec$.
 
In the first case, $\crec\le \frac{2n-4}{n^2}$, 
The cross point is at $\Ddot{v}=\frac{\sqrt{1+4\crec(n -1)}-1}{2\crec}$, therefore 
\begin{align}
    & \dct(\msn,\mathcal{M}) \le \dctmix(\stalgmsn,\spalgmsn,M(\Ddot{v})) \nonumber =\\&
     =1+\frac{\sqrt{1+4\crec(n -1)}-1}{2\crec}\crec
\end{align}

For the second case, where $\crec > \frac{2n-4}{n^2}$, \msn will cross the path of \upper algorithm where \upper is using $\rrbvn$, that is, where the DCT is at $2-\frac{2}{n}$. In this case, the upper bound is the same value as \rsn, but at a different $v$ value, we denote this value as $\Dot{v}=\frac{n-2}{n\crec}$.
 
\begin{align}
   & \dct(\msn,\mathcal{M}) \le \dctmix(\stalgmsn,\spalgmsn,M(\Dot{v}))=\nonumber\\&
   =2-\frac{2}{n}
\end{align} 
This finalizes the proof.
\end{proof}

Combining the result from \autoref{app:worstCaseAnalysisMVU}, that the worst case matrix for $M(v,u)$ for all systems is at $u=0$ with our recently proven \autoref{app:thm:DCTMixMvuFUll}, and our results regarding \csn, form \autoref{obs:main:DAsysDCTBound}  and \rsn from \autoref{thm:main:RRsysDCT} for the system DCT of both systems, we prove our main theorem of interest, that is \autoref{thm:mainSystemDCTreslut}.

\subsection{Theoretical \msn Using the Lower Bound of \rsn }\label{sec:loweroundOfRSN}
Let us denote \lowerRsn~ as a traffic scheduling algorithm that achieves the lower bound for \rsn on $M(v)$ as we found in \autoref{cor:lowerBoundRSNMV}. We make no claims here about whether this algorithm is achievable for every $M(v)$.
The lower for \rsn with $M(v)$ is the following
 \begin{align}
     \dctrot(\rr,\lowerRsn, M(v) \ge \frac{2(n-1)}{n-1+v} 
\end{align} 

We may now ask at which value of $v$ does the curve of \csn crosses the lower bound of \rsn, that is $\lowerRsn$.

\begin{align}
     &\dctrot(\rr,\lowerRsn, M(v))= \dctda(\bvn,\direct, M(v))=\\
&=2-\frac{2v}{n-1+v}= 1+\crec v=\\
 &=\crec v^2+v(1+\crec (n-1))-n+1=0
\end{align}
We can solve for $v$ by solving a quadratic equation.
\begin{align}
    v=\frac{\crec (1-n)-1 \pm \sqrt{4\crec(n-1)+(\crec -1 -n \crec)^2}}{2\crec}
\end{align}
Only the positive sign solution gives a relevant solution where $v>0$ .
Let us denote this value as $\textsc{low}(\crec,n)$.
$$\textit{low}(\crec,n)=$$
$$=\frac{\crec (1-n)-1 + \sqrt{4\crec(n-1)+(\crec -1 -n \crec)^2}}{2\crec}$$

\begin{gather*}
   \dctmix(\stalg,\spalg,M(v)) = \\
\begin{cases}
  1+\crec v & \text{for }1\le v<\text{low}(\crec,n)\\
  \frac{2(n-1)}{n-1+v}   & \text{otherwise}
\end{cases}
\end{gather*}

\section{Improving the Lower DCT Bounds on M(v) for {\rsn}}\label{sec:lowerDCTBoundRSNMC}
In this section, we discuss the lower bounds for \rsn for the particular case of the $M(v)$ matrix family we defined in \autoref{subsec:MvuMatrixDef}. Recall that in this doubly stochastic matrix family, the demand matrix is composed of the summation of a set of $v$ equal and non-overlapping permutation matrices. This family, as we will later see, presents both the worst and best-case matrices for \rsn, and thus, we believe it allows us to explore a wide range of possible performances for the \rsn type systems. While the bounds we present here are not strictly necessary for the analyses of the system DCT, we present this section as an independent contribution to the paper.

Recall the lower bound for \rsn on $M(v)$ in \autoref{cor:lowerBoundRSNMV}. We note that the question of 
what is a tight lower bound for any matrix is still an open question. 
However, we believe that this lower bound is achievable by some algorithm for all values of $v$, where $2\le v \le n-2$, but only for some configurations of permutations in $M(v)$.  
In the following section, we show that our lower bound in not tight even for $M(2)$. We show that there exists some matrices in $M(2)$ which do not allow DCT lower than the upper bound for any \rsn type packet schedule. These matrices have propriety we denote as a dual collision as in \autoref{def:collisionDual}.

\subsection{Bounds For The M(2) Matrix}\label{app:sec:MtwoBounds}
In this section, we discuss bounds to the $M(2)$ family of matrices. Recall that according to \autoref{cor:lowerBoundRSNMV}, the lower bound for $M(2)$ is $2-\frac{4}{n+1}$. In the following theorem we prove a better bound for some cases. 

\begin{restatable}{theorem}{lowerBoundvTwo}[Lower bound of $M(2)$ with a dual collision]\label{thm:lowerBoundv2}
For \rsn there exists a demand matrix $M \in M(2)$ that has, at least, two duel collision cells such that DCT for $M$ lower bounded by the upper bound in \autoref{thm:upperRSNBVN} for any traffic scheduling $\spalg$. That is;
\begin{align}
 \dctrot(\rr, \spalg, M &\geq (2-\frac{2}{n})
\end{align}
 \end{restatable}

We note that this is the same DCT given by \rrbvn.
This theorem shows that for some matrices from $M(2)$, the upper bound is optimal. However, we also present an example of a matrix from this family and a traffic schedule where the lower bound us reached.  
We, therefore, conjecture that the optimal DCT on \rsn on any matrix in $M(v)$ is found between the upper and lower bounds as described in this section. 

For the following proof of \autoref{thm:lowerBoundv2}, we must first define a \emph{collision cell}. This is a cell in the demand matrix that has a collision between direct non-local traffic and direct local traffic two-hop routing. These collisions might limit what a routing algorithm could do in some instances.
A collision occurs when there is a demand to send two 'types' of traffic from the same cell simultaneously. We are only interested when these two types of traffic are direct non-local and direct local traffic. Furthermore, although this can be generalized, we define our collision for matrices from the matrix family $M\in M(v)$.  

We are only interested in the relation between cells with demand and cells with zero demand in the demand matrix. That is, a collision is caused by a certain non-zero demand cell sending indirect traffic, which, on its second hop, collides with traffic from another non-zero demand cell. 
Our definition of a \emph{collision cell} is the following.
 
\begin{definition}[Collision Cell]\label{def:collisionSing}
Given a demand matrix $M\in M(v)$, and a cell with demand $M[l,a]>0$, an empty cell $M[l,j]=0$ is a collision cell if its second hop when helping the cell $M[l,a]$, that is $M[j,a]$ also has any demand. Or more formally, the cell $M[l,j]$ is a collision cell, for the cell $M[l,a]$ if the following conditions holds
\begin{align}
     M[l,a]>0 \And M[l,j]=0 \And M[j,a]>0
\end{align}
\end{definition}
The previous definition relates to a collision cell caused by a single other cell. An expansion of this definition is a \emph{dual collision cell}.

\begin{definition}[Dual Collision Cell]\label{def:collisionDual}
Given a demand matrix  $M\in M(v)$, and two cells with demand $M[l,a]>0$,$M[l,b]>0$, an empty cell with the coordinates $M[l,j]$ is a dual collision cell if it has a collision both with $M[l,a]$ and $M[l,b]$. Or, more formally, $M[l,j]$ is a dual collision if the following conditions holds
\begin{align}
    M[l,a]>0 \And M[l,b]>0 \And M[l,j]=0\And \nonumber \\ \And M[j,a]>0 \And M[j,b]>0
\end{align}
\end{definition}
Our observation here is that dual collision cells often come in pairs.
\begin{observation}[Pairs of Dual Collision Cell]\label{obs:PaircollisionDual}
   If $M[l,j]$ is a dual collision cell in relation to $M[l,a]$, $M[l,b]$, then if $M[j,l]=0$, it is also a dual cell in relation to  $M[j,a]>0 \And M[j,b]$.
\end{observation}
When $M[j,l]$ is also zero, we can see [that]] it also must fulfill the condition $ M[j,a]>0 \And M[j,b]>0 \And M[j,l]=0 \And M[l,a]>0 \And M[l,b]>0$. The case $M[j,l]=0$  will be true for any matrix $M \in M(2)$ since it has exactly two non-zero cells per line.

To find a matrix $M\in M(2)$ with two dual collision cells, we need to find a matrix with two rows with the same two columns with non-zero demand.

We build this matrix $M$, from two permutation matrices, $P_1$ and $P_2$, where $M=0.5(P_1+P_2)$. We can set $P_1[1,2]=1$, $P_1[n,n-1]=1$ and $P_2[1,n-1]=1$ ,$P_2[n,2]=1$.
Other elements can be chosen arbitrarily, keeping with the $M(v)$ conditions as we described in \autoref{sec:thecasestudy}.
In this case, the cells $M[1,n]$ and $M[n,1]$ are dual collision cells.

After stating our definitions, we can prove  \autoref{thm:lowerBoundv2}
 \begin{proof}[Proof of \autoref{thm:lowerBoundv2}]
The first part resembles our proof from \autoref{thm:lowRSNApp}. 
Consider a complete schedule from \autoref{app:clm:TotalmatrixDCT}.  
The lower bound for the DCT is minimized when the maximum cell in $\totalDemandMat$ is minimized. 
Hence, in this proof we focus on rows with a duel collision cell, as even if some rows might allow for a lower DCT, rows with duel collision cells will dominate the demand completion time.
Next, we claim that since the schedule is optimal, each cell in 
$M$ that has an entry $\frac{1}{2}$ will only send \emph{direct} traffic. This follows from the observation that if it was needed to send an \emph{indirect} traffic that travels two hops, it was better to send the same amount of traffic directly, reducing the overall traffic sent. Contradicting the optimality of the schedule.

Looking at the dual collision cells, there are only two fundamental ways in which any algorithm $\spalg$ can use the dual collision cells in the \rsn scheme: either it is not used, or it is used to send indirect traffic.

For the first case, where $\spalg$ leaves these cells unused, the dual collision cells are inactive cells.
To find the lower bound on DCT, in this case, we use the upper bound $\skew$ we found in \autoref{thm:lowRSNApp} and using the expression for DCT we found in \autoref{thm:lowerBoundRSN}, for any algorithm $\spalg$ with $\inActCellMax=2$ we find that the following bound.
\begin{align}
    & \dctrot(\rr,\spalg, M\geq \nonumber\\ &\geq (2-\frac{4}{n } )\frac{n-1}{n-2}\frac{1}{ \eff r} \frac{ \Wcard{M}}{n}= (2-\frac{2}{n})\frac{1}{ \eff r} \frac{ \Wcard{M}}{n}
\end{align}
Which is the same bound as the general \emph{upper} bound for $M(2)$ we found in \autoref{thm:upperRSNBVN} by using the algorithm $\rrbvn$.

For the second case, $\spalg$ \emph{uses} the dual collision cells to send indirect traffic.
We denote the two permutations which compose the matrix $M$ as $P_1$ and $P_2$.
The algorithm uses use the dual collision cells to send $x_1$ bits from $P_1$ and $x_2$ bits of data from $P_2$ (per row), where $0\leq x_1, x_2\leq 0.5$.
Note that sending data from the cells of $P_1$ to cells of $P_2$ via the dual collision cells means that any data sent from the cells of $P_1$ to these cells will end up being transmitted on connections that could have been used to transmit data from the cells of $P_2$ directly (and vice versa for $P_2$). That is, if $P_2[k,j]>0$ then a bit sent to a dual collision cell would eventually have to be sent on the connection $(k,j)$, instead of direct local packets from $P_2$.
This means that by using the collision cells, we can send fewer bits directly from each permutation in the same time frame.

To prove that using the dual collision cells doesn't help to lower the DCT we prove that the following holds. 
First we establish the amount of data sent in the base case, where the dual collision cells are inactive cells, that is, $x_1 = x_2 =0$. Let $\totalDemandMat$ be a total traffic matrix formed from an optimal traffic algorithm that doesn't use the dual collision cell. Let $\card{P_1^T}$ and $\card{P_2^T}$ be the total demand sent on connections that are equivalent to $P_1$ and $P_2$ in the total traffic matrix.

If the the cell of permutation $P_1$ transfer $x_1$ bits  
and the cell of $P_2$ transfer $x_2$ bits to the dual collision cell (of each row with a duel collision), then the total size of data transmitted on the connections, which are equivalent to $P_1$ is now $\card{c_2^T}-x_2+x_1$ and for $P_2$ is now $\card{c_2^T}-x_1+x_2$. Since the cells use the dual collision cells, any demand reduced from a cell in $P_1$ is transferred to a cell in $P_2$ (or vice versa). Lastly, we will assume w.l.o.g that $x_1\geq x_2$, and we can conclude that $\card{P_2^T}-x_1+x_2\geq \card{P_2^T}$, that is, the new transmission time is greater or equal to the original transmission time before using the dual collision cell for any value of $x_1$ and $x_2$. 
Since we haven't been able to reduce the size of connections on the permutations, the size of the other connections is irrelevant. 
In conclusion, the total DCT for any algorithm $\spalg$ cannot be less than the bound we found for $\rrbvn$. That is   
\begin{align}
&\dct (\rr,\spalg,M)\geq \dct (\rr,\rrbvn,M(2))=\\&=(2-\frac{2}{n})\frac{1}{ \eff r} \frac{ \Wcard{M}}{n}
\end{align}
Since we have already proven that this bound is possible, the theorem follows.

 \end{proof}
 
\begin{figure*}[ht]
  \begin{centering}
  \begin{tabular}{c}
 \includegraphics[width=1\textwidth]{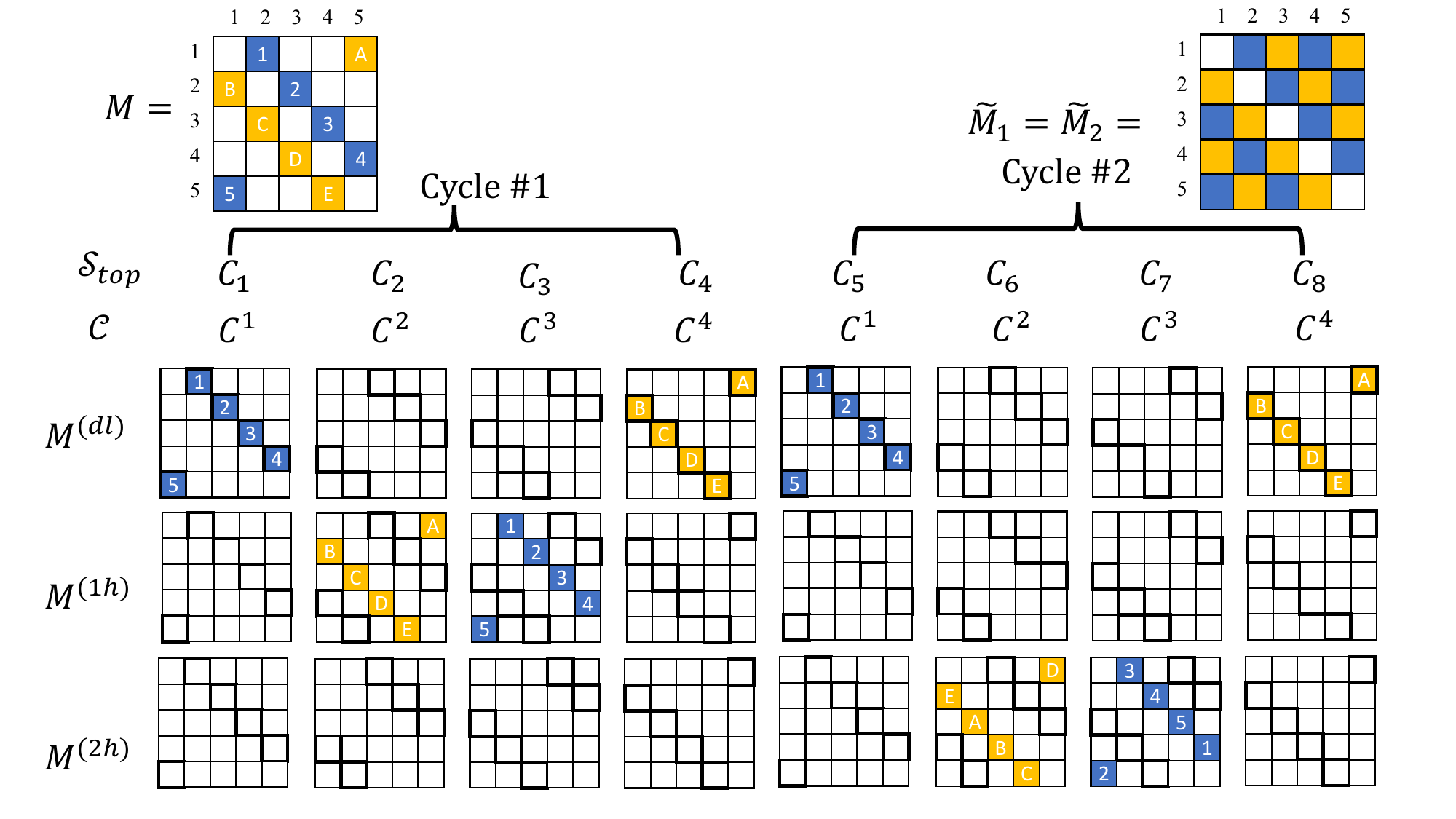}  
  \end{tabular}
   \caption[A brief example of a schedule for a small matrix from $M(2)$ which achieves the lower DCT bound]{An example of a two-stage traffic schedule which reaches the lower bound on a $5\times 5$ demand matrix $M$, and the traffic schedule generated for it, \ts. The traffic schedule $\ts=\{\configElem_1,\configElem_2,\configElem_3,\configElem_4\}$ is the same topology schedule as used in \autoref{fig:example}. The permutation matrices composing $M$ are marked with yellow and blue colors.  
   }
    \label{fig:5on5LowerBoundSched}
  \end{centering}
 \end{figure*}

\subsection{Possible Optimal Traffic Schedule Example}\label{app:PossibleLowerBoundMTwo}
The proof of \autoref{thm:lowerBoundv2} shows that the upper bound we found is tight not just for all matrices from $M(1)$ and $M(n-1)$ but for also for some matrices from $M(2)$, that contains a dual collision cell. However, we note that this criterion may not be the only one. 
We now show an example for a matrix $M\in M(2)$, which doesn't contain a dual collision but two (singular) collisions instead.

In \autoref{fig:5on5LowerBoundSched}, we show a traffic schedule ($\ps=\{\Mdl,\Mth,\Moh\}$) presenting a possible scheduling achieving the lower bound we found in \autoref{cor:lowerBoundRSNMV}. Here the topology scheduling used is the same as in \autoref{fig:example}.
Looking at $M\in M(2)$, in the figure, we see that it is the sum of two permutations matrices, the one in blue marked with numbers ($1$ thorough $5$), and the other in yellow marked with letters ($A$ thorough $E$). These color markings are persistent throughout the figure.  
In the indirect local traffic set, $\Moh$, we see how the traffic from each matching is transferred to a different cell. While in the direct local set  $\Mth$, we see the destination of each cell in the indirect first hop traffic.  Similarly to schedules algorithm $\algOnePerm$ this schedule has two cycles, with two identical total traffic matrices in each cycle $\totalDemandMat_1$ and $\totalDemandMat_2$, each being uniform.
In total, these would give us a traffic schedule that achieves the lower bound where each matching cell has a normalized load of $\frac{1}{2}$ in $M$ and a load of $\frac{1}{6}$ during the schedule. Finally, we set $\delta=\frac{1}{6}$ and see how this schedule is viable.
The DCT of this schedule is $ 2\cdot4\cdot\frac{1}{6}=\frac{4}{3}\approx1.34 [sec]$, that is, it the slot duration $\delta=\frac{1}{6}$ times the number of slots in a cycle, $n-1=4$ times the number of cycles which is $2$. While if we used $\rrbvn$, the result for this matrix $M$ would have been $2-\frac{2}{5}=1.6 [sec]$, and this is clearly higher. 
Note that the topology schedule we presented here in \autoref{fig:5on5LowerBoundSched} contains the two permutations from $M$. However, this is not a requirement, and we could generate a schedule with the same DCT for another topology schedule.  
 
Following the proof and lower bound example, we believe that all matrices from $M(v)$ have specific cases, that is, arrangements of the $v$ matchings, where the matrix DCT is at the lower bound and others where it is at the upper bound.
and that this depends on the number of collisions. 
We leave proof of this for future work.

 \end{document}